\documentclass{sig-alternate}
\usepackage{amsmath}
\usepackage{txfonts}
\usepackage{amssymb}
\usepackage{times}
\usepackage{graphicx}
\usepackage{epsfig,tabularx,amssymb,amsmath,subfigure,multirow}
\usepackage[linesnumbered,ruled,noend]{algorithm2e}
\usepackage[noend]{algorithmic}
\usepackage{multirow}
\usepackage{graphicx,floatrow}
\usepackage{listings}
\usepackage{threeparttable}
\usepackage{tikz}
\usepackage[T1]{fontenc}
\usepackage{pgfplots}
\newcommand{\nop}[1]{}
\newcommand{\tabincell}[2]{\begin{tabular}{@{}#1@{}}#2\end{tabular}}
\newtheorem{definition}{Definition}[section]
\newtheorem{theorem}{Theorem}[section]

\newtheorem{example}{Example}

\begin{document}
%
\conferenceinfo{WOODSTOCK}{'97 El Paso, Texas USA}

\title{On The Marriage of SPARQL and Keywords}
%
%
%
%
%

\numberofauthors{3} 
%

\nop{
\author{
%
%
\alignauthor
Peng Peng\\
       \affaddr{Peking University}\\
       \affaddr{Beijing, China}\\
       \email{pku09pp@pku.edu.cn}
\alignauthor
Lei Zou\\
       \affaddr{Peking University}\\
       \affaddr{Beijing, China}\\
       \email{zoulei@pku.edu.cn}
\alignauthor Dongyan Zhao\\
       \affaddr{Peking University}\\
       \affaddr{Beijing, China}\\
       \email{zhaodongyan@pku.edu.cn}
}
}

\author{%
{{Peng Peng}, Lei Zou{}, Dongyan Zhao{}}%
\vspace{1.6mm}\\
\fontsize{10}{10}\selectfont\itshape Peking University,
China;\\
\fontsize{9}{9}\selectfont\ttfamily\upshape $\{$
pku09pp,zoulei,zhaodongyan$\}$@pku.edu.cn
\\}

\maketitle

\begin{abstract}
Although SPARQL has been the predominant query language over RDF graphs, some query intentions cannot be well captured by only using SPARQL syntax. On the other hand, the keyword search enjoys widespread usage because of its intuitive way of specifying information needs but suffers from the problem of low precision. To maximize the advantages of both SPARQL and keyword search, we introduce a novel paradigm that combines both of them and propose a hybrid query (called an SK query) that integrates SPARQL and keyword search. In order to answer SK queries efficiently, a structural index is devised, based on a novel integrated query algorithm is proposed. We evaluate our method in large real RDF graphs and experiments demonstrate both effectiveness and efficiency of our method.
\end{abstract}

\section{Introduction}
\label{sec:introduction}

As more and more knowledge bases become available, the question of how end users can access this body of knowledge
becomes of crucial importance. As the de facto standard of a knowledge base, RDF (Resource Description Framework) repository is a
collection of triples, denoted as $\langle$ subject, predicate, object $\rangle$. An RDF repository can be represented as a graph, where subjects and objects are vertices connected by labeled edges (i.e., predicates).  Figure \ref{fig:exampleRDF} shows an example RDF dataset and the corresponding RDF graph, which is a part of a well-known knowledge base Yago \cite{DBLP:yago}. All subjects and objects correspond to vertices and predicates correspond to edge labels. The numbers besides the vertices are IDs, and they are introduced for the easy presentation.

\begin{figure*}%
\begin{minipage}{0.48\textwidth}
   \centering
   \subfigure[][{\scriptsize Example RDF Triples}]{%
      \includegraphics[scale=0.27]{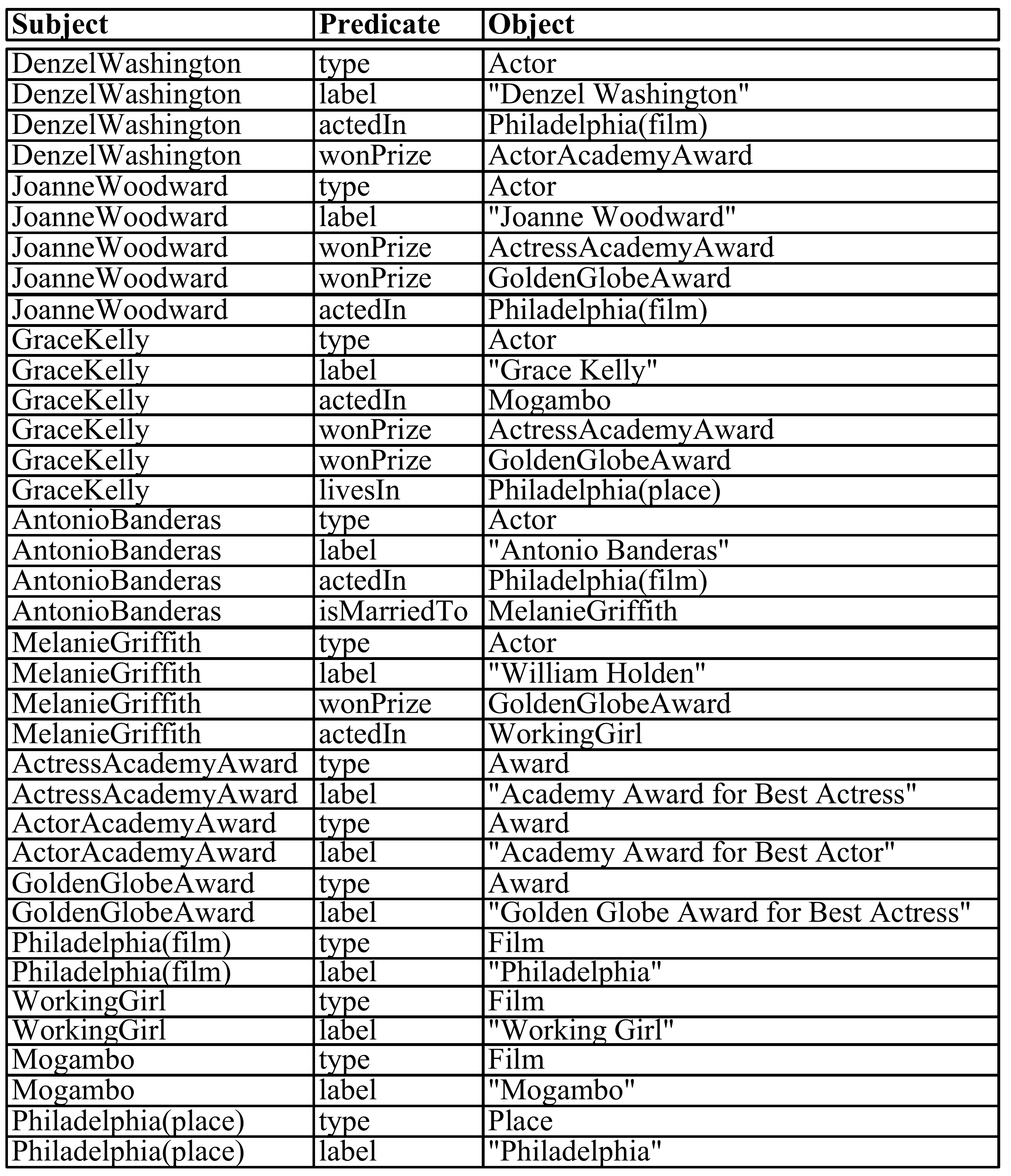}
       \label{fig:exampleRDFTriples}%
       }%
       \end{minipage}
        \begin{minipage}{0.48\textwidth}
       \subfigure[][{\scriptsize Example RDF Graph}]{%
      \includegraphics[scale=0.23]{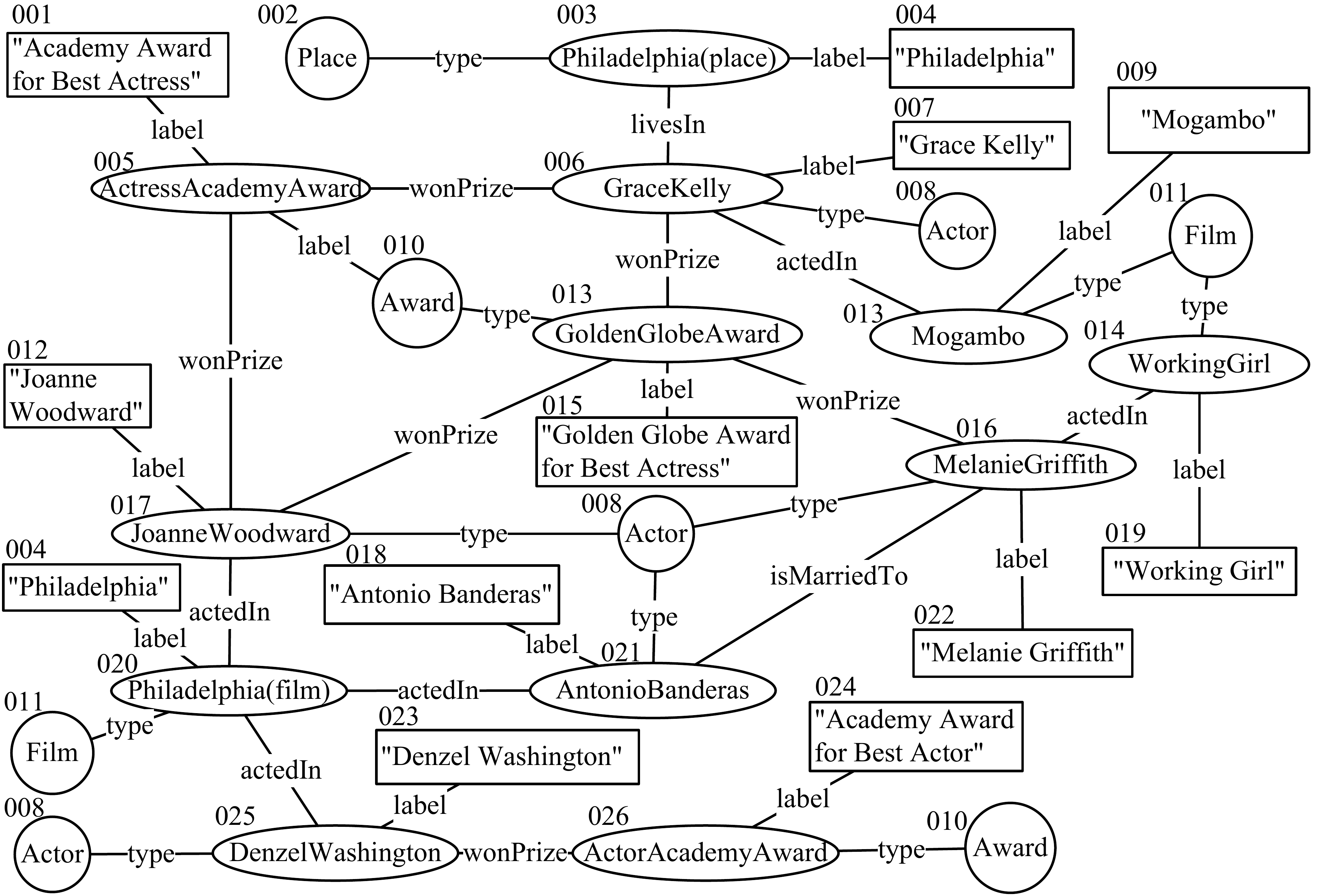}
       \label{fig:exampleRDFGraph}%
       }%
       \\
       \subfigure[][{\scriptsize Example SPARQL-Keyword Query}]{%
      \includegraphics[scale=0.28]{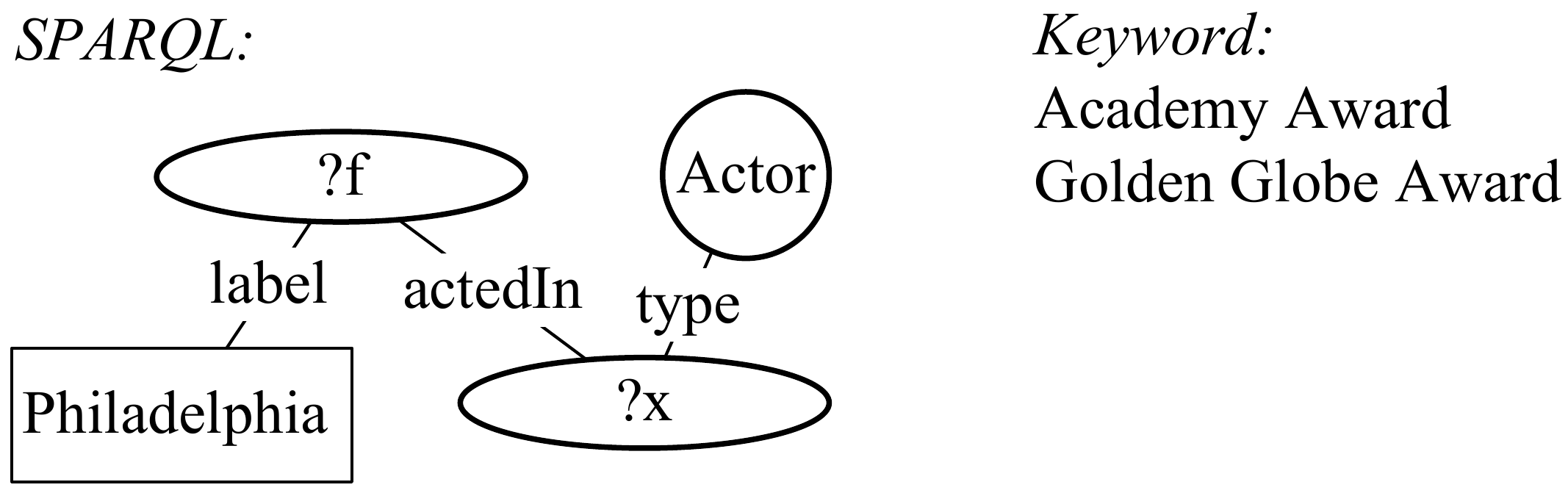}
       \label{fig:ExampleSPARQLKeyword}%
       }%
       \end{minipage}
       \vspace{-0.3in}
\caption{ RDF Examples}%
 \label{fig:exampleRDF}
\end{figure*}

As we know, SPARQL query language is a standard way to access RDF data and is based on the subgraph (homomorphism) match semantic \cite{DBLP:journals/tods/PerezAG09}. Figure \ref{fig:sparqlexample_1} shows an example of SPARQL query and its corresponding query graph graph is shown in Figure \ref{fig:ExampleSPARQLKeyword}.  The query semantics of the example SPARQL is ``finding all actors starring in film \emph{Philadelphia}'',  To enable to use SPARQL, users should have full knowledge of the whole RDF schema. For example, users should know that predicate ``actedIn'' means ``starring in'' and the Philadelphia film's URI is ``Philadelphia(film)''. In real applications, it may not be practical to have full knowledge about the whole schema; thus, it may not be possible to specify exact query criteria. The following example illustrates the challenges.

\begin{example} \label{example:1} Find all actors starring in film \emph{Philadelphia}, who are related to ``Academy Award'' and ``Golden Globe Award''. Assume that we do not know the exact URIs corresponding to ``Academy Award'' and ``Golden Globe Award''. Furthermore, there is no precise predicate corresponding to ``related to''.
\end{example}

\vspace{-0.3in}
\begin{figure}[h]%
   \centering
   \subfigure[][{\scriptsize $Q_1$}]{%
      \includegraphics[scale=0.3]{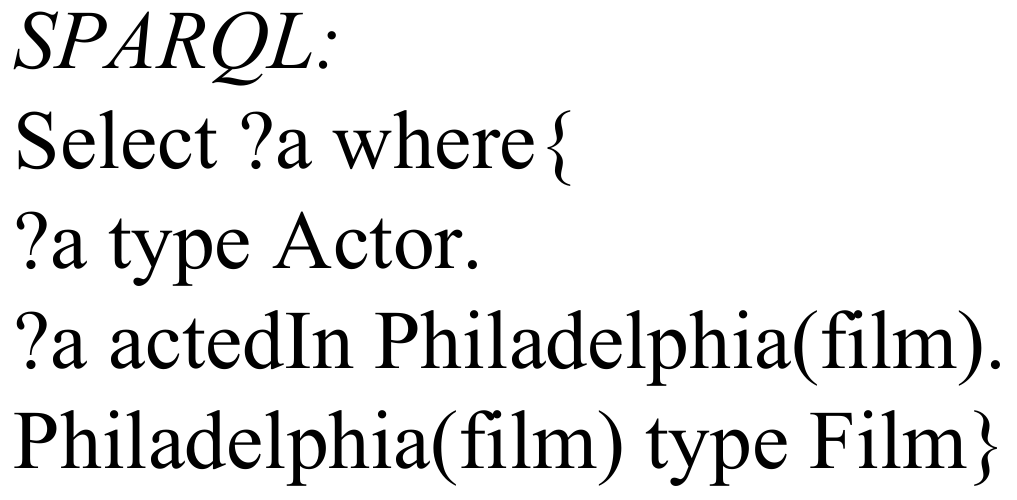}
       \label{fig:sparqlexample_1}%
       }%
   \subfigure[][{\scriptsize $Q_2$}]{%
      \includegraphics[scale=0.3]{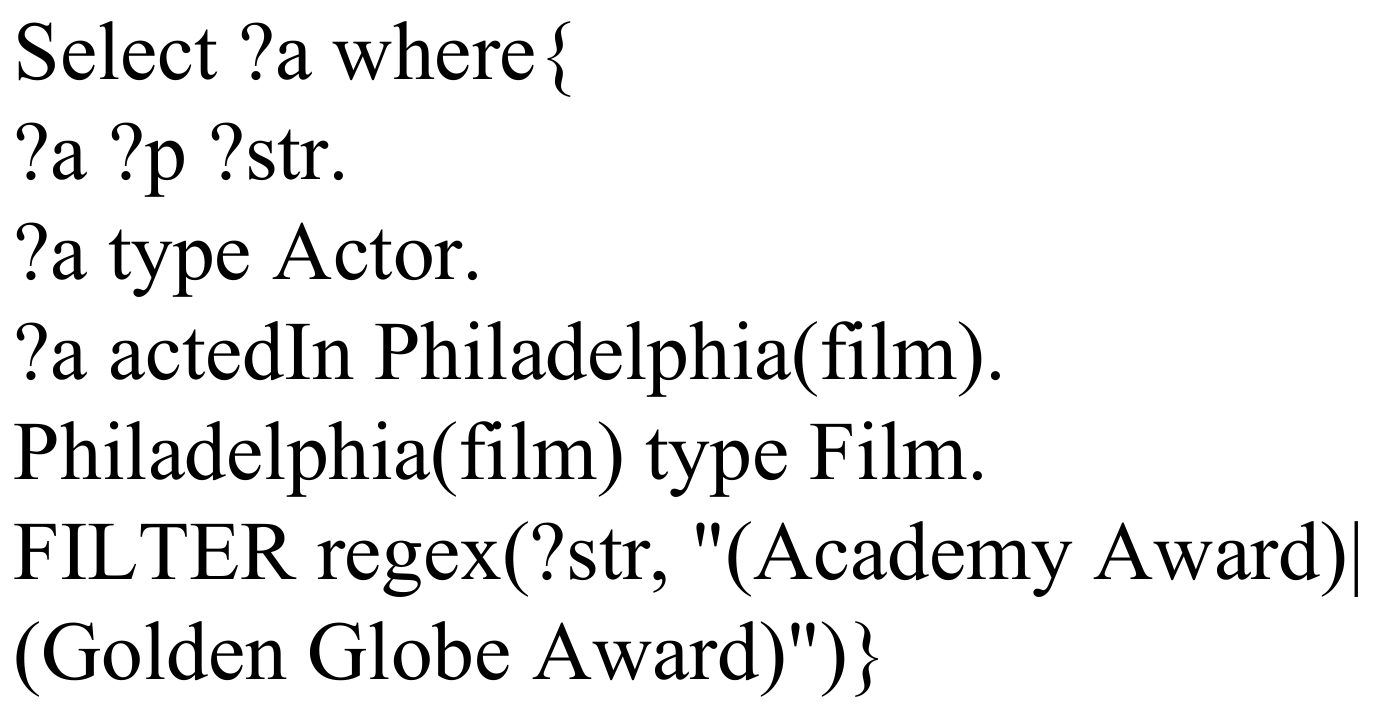}
       \label{fig:sparqlexample_2}%
       }%
        \vspace{-0.25in}
 \caption{ Example SPARQL Queries}%
 \label{fig:sparqlexample}
\end{figure}
\vspace{-0.1in}

There are two issues in this example. First, since we do not know the URIs of ``Academy Award'' and ``Golden Globe Award'', we should provide a keyword search paradigm that maps the keywords to the corresponding entities or classes in RDF graphs. Existing SPARQL syntax only supports the regular expression, as shown in Figure \ref{fig:sparqlexample}. More typographic or linguistic distances, such as string edit distance \cite{DBLP:conf/sigmod/DengLF14} and google similarity distance \cite{DBLP:GoogleDistance}, are desirable.

The second issue is that there is no precise predicate corresponding to ``related to''. A possible solution is to use ``unknown'' predicate (i.e., a variable at the predicate position), but, it only finds one-hop relations. Figure \ref{fig:sparqlexample_2} shows a SPARQL query with unknown predicate and the regular expression FILTER. However, it fails to finding the multiple-hop relations, which may also be informative to users. For example, Antonio Banderas, an actor starring in Philadelphia film, whose wife won ``Golden Globe Award''.  This is also a possible interesting result to users, but, this is a two-hop relation between Antonio Banderas and ``Golden Globe Award''.

In contrast, keyword search \cite{DBLP:BANKS,DBLP:Bidirectional,DBLP:BLINKS,DBLP:DPBF,DBLP:star,DBLP:rcliques} on graphs provides an intuitive way of specifying information needs. For example, we input two keywords ``Joanne Woodward'' and ``Golden Globe Award'' to discover unbounded relations (i.e., the paths in RDF graphs) between them. However, keyword search may return a larger number of non-informative search answers to users.

In fact, users' query intensions cannot be well modeled using a single query type in many real-life applications. Hence, a hybrid search capability is desired. In this paper, we propose an integrated query formulation (called a \emph{S}PARQL-\emph{K}eyword query, shorted as SK query) and the solution framework by combining advantages of SPARQL and keyword search. Generally speaking, the results of SK query is the $k$ SPARQL matches that are closet to all keywords in RDF graph $G$, where $k$ is a parameter given by users. The formal definition of SK query is given in Definition \ref{def:query}.

Let us recall Example \ref{example:1} again. We issue the following SK query $\langle Q,q\rangle$. The SPARQL query graph $Q$ is given in Figure \ref{fig:ExampleSPARQLKeyword}, while the keywords are $q=\{$Academy Award, Golden Globe Award$\}$. Figure \ref{fig:yagoresultofsk} shows three different results. First, there are three different subgraph matches of query $Q$, i.e., $M_1$, $M_2$ and $M_3$. Then, the keywords are matched in different literal vertices, i.e., $001$, $015$ and $026$. The distance between a subgraph match $M$ and a keyword in $q$ is the shortest distance between $M$ and one vertex containing keywords. We find that $M_1$ is the closest to the two keywords. It says ``Joanne Woodward starring in Philadelphia film won both Academy Award and Golden Globe Award''. Obviously, this is an informative answer to the query in Example \ref{example:1}.

\begin{figure}[h]
 \vspace{-0.1in}
\begin{center}
    \includegraphics[scale=0.21]{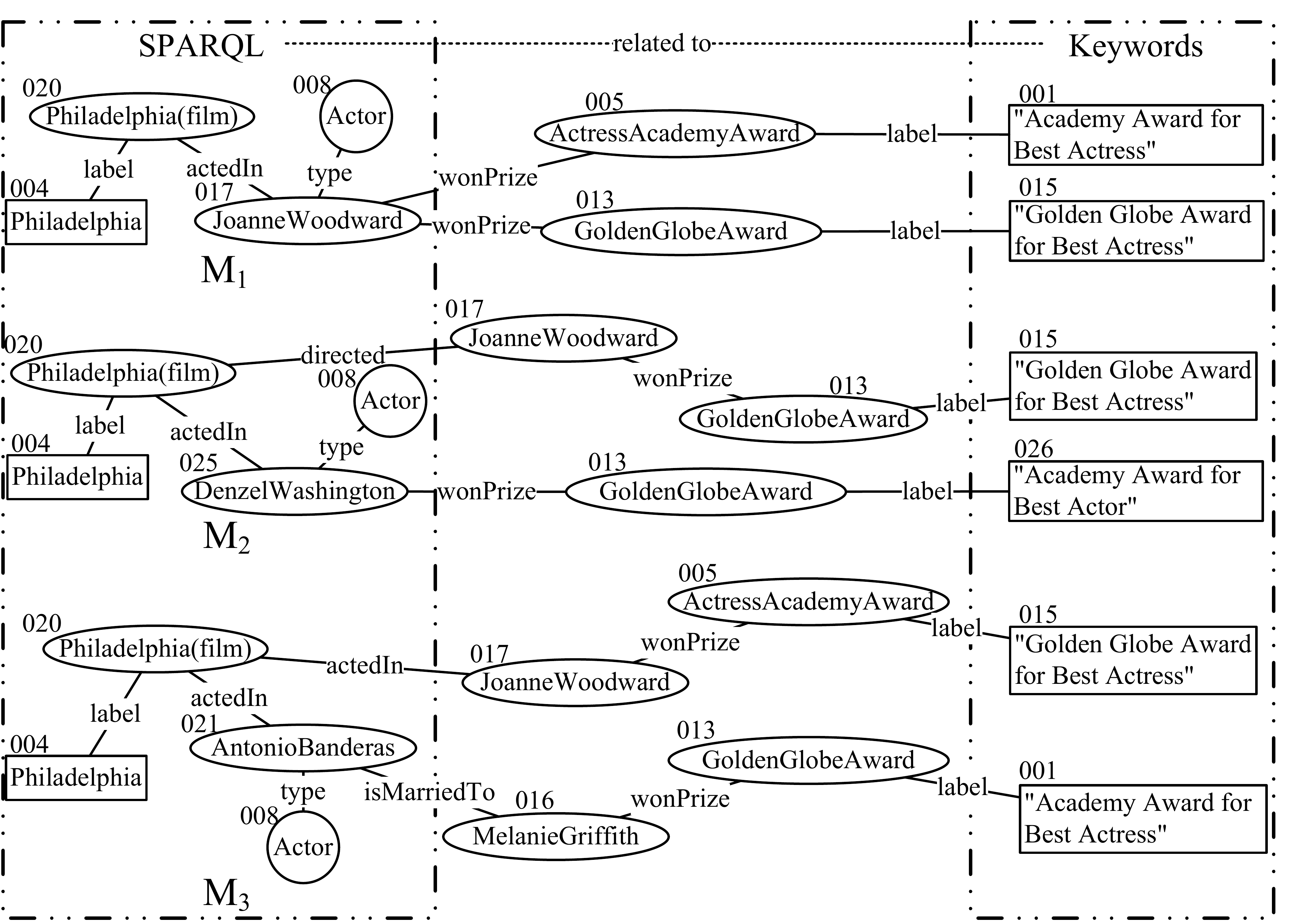}
    \vspace{-0.25in}
   \caption{ SK Query Results}
   \vspace{-0.25in}
   \label{fig:yagoresultofsk}
\end{center}
\end{figure}
\vspace{-0.1in}

In the above analysis, we assume that the relation strength depends on the path length, i.e., the number of hops. Actually, different predicates should have different weights to evaluate the relation strength. For example, there are two two-hops paths from 021 (AntonioBanderas) to 017 (JoanneWoodward). The first one is through 020 (Philadelphia(film)), while the second is through 008 (Actor). It is obvious that two people cooperating in the same film is more meaningful than both of them being actors, so the two-hops path through 020 has more relation strengths than the two-hops path through 008. Therefore, following the intuition of TF-IDF for measuring the word importance in a corpus, we propose ``predicate salience'' (see Section \ref{sec:background}) to evaluate relation strengths.

Another challenge of this problem is the search efficiency. A naive ``\emph{exhaustive-computing}'' strategy works as follows: we first find all subgraph matches of $Q$ (in RDF graph $G$) by existing techniques; then, we compute the shortest path distances between these subgraph matches and the vertices containing keywords on the fly; finally, the matches with shortest distances to keywords are returned as answers. Obviously, this is an inefficient solution. Given a SPARQL query $Q$, there may exist some matches of $Q$ that are far from the keywords in RDF graphs. These matches cannot contribute to the final results. Therefore, it is unnecessary to identify all subgraph matches in RDF graph. Instead of the exhaustive computing, we only find matches of SPARQL query $Q$ progressively and design a lower bound that stops the search process as early as possible. Moreover, we propose a star index to enable the structural pruning.

In summary, we made the following contributions in this paper.

\begin{enumerate}
\item We propose a new query paradigm over RDF data combining keywords and SPARQL(called an SK query), and design a novel solution for this problem.

\item We design an index to speed up SK query processing. We propose a frequent star pattern-based index to reduce the search space.

\item We evaluate the effectiveness and efficiency of our method in real large RDF graphs and conclude that our methods are much better than comparative models in both effectiveness (in terms of NDCG$@k$) and query response time.
\end{enumerate}

The remainder of this paper is organized as follows: Section \ref{sec:background} defines the preliminary concepts. Section \ref{sec:overview} gives an overview of our approach. We introduce a structural index to efficiently find the candidates of variables in SPARQL queries in Section \ref{sec:index}. We discuss how to compute the results of SK queries in Section \ref{sec:topkcomputation}. Experimental results are presented in Section \ref{sec:experiments}. Related works and the final conclusion are drawn in Section \ref{sec:related} and \ref{sec:conclusions}, respectively.

\vspace{-0.075in}
\section{Background}
\label{sec:background}
In this section, we introduce the fundamental definitions used in this paper.

\vspace{-0.075in}
\subsection{Preliminaries}
\label{sec:Preliminaries}

An RDF dataset consists of a number of triples, which is corresponding to an RDF graph. The SK query is to find $k$ SPARQL matches that are top-k nearest with regard to all keywords.
\nop{
Table \ref{table:notations} lists commonly used notations in this paper.

\begin{table}[t]

\centering
\begin{tabular}{|p{3.35cm}||p{4.8cm}|}
\hline
   \textbf{Notations} &  \textbf{Meanings} \\ \hline \hline
  $G=(V(G),E(G),L)$ &  RDF Graph \\
  \hline
  $\langle Q,q=\{w_1,w_2,...,w_n\}\rangle$    & A \emph{S}PARQL $\&$ \emph{K}eyword query, where $Q$ is the SPARQL and $q$ is the list of keywords.\\
  \hline
  $r=\langle M,\{v_1,v_2,...,v_n\}\rangle$    & A result of the SK query $\langle Q,q \rangle$, where $M$ is a match of $Q$ in $G$ and $v_i$ is a vertex (in $G$) containing keyword $w_i$ in $q$.  \\
  \hline
  $Cost_{content}(r)$    & The content cost of a result $r$ \\
    \hline
  $ Cost_{structure} (r)$    & The structure cost of a result $r$ \\
  \hline
  $d(M,v)$    & The distance between a match $M$ and a vertex $v$   \\
  \hline
  $d(v,w_i)$    & The distance between a vertex $v$ and a keyword $w_i$   \\
 \hline
 $PQ_i$    & A priority queue for keyword $w_i$. \\
  \hline
 $(v, p,|p|)$    &  A basic element in $PQ$, where $v$ is a vertex, $p$ is a path between $v$ and some vertices containing keywords and $|p|$ denotes the path distance. \\
 \nop{
  \hline
 $t$    & A length-m vector, where the $j$-th dimension in $t$ is a vertex $v$ corresponding to variable $u_i$ in SPARQL query $Q$. \\
 }
 \hline
 $RS_i$    & A result set for keyword $w_i$. \\
 \hline
 $d[v]$    & A length-n vector for vertex $v$. If $v\in RS_i$, $d[v][i]=d(v,w_i)$; otherwise, $d[v][i]=null$. \\
 \hline
  \end{tabular}
\label{table:notations}
\caption{Notations}
\end{table}
}

\begin{definition}
\label{def:rdfdatagraph}
An RDF data graph $G$ is denoted as $\langle V(G), E(G),$ $L\rangle$, where (1) $V(G)=V_L\cup V_E \cup V_C$ is the set of vertices in RDF graph $G$, where $V_L$, $V_E$ and $V_C$ denote literal vertices, entity vertices and class vertices, respectively; (2) $E(G)$ is the set of edges in $G$; and (3) $L$ is a finite set of edge labels, i.e. predicates.
\end{definition}

\begin{definition}\label{def:query}

A \emph{SK (\emph{S}PARQL \& \emph{K}eyword)} query is a pair $\langle Q, q\rangle$, where $Q$ is a SPARQL query graph and $q$ is a set of keywords $\{w_1, w_2, ... ,w_n\}$.

Given an SK query $\langle Q, q\rangle$, the \emph{result} of $\langle Q, q\rangle$ in a data graph $G$ is a pair $\langle M, \{v_1, v_2, ... ,v_n\}\rangle$, where $M$ is a subgraph match of $Q$ in $G$ and $v_i$ ($i=1,...,n$) is a literal vertex (in $G$) containing keyword $w_i$.
\end{definition}

\nop{In practice, the SK query requires returning not only the distances but also the paths between SPARQL matches and keywords. We only need simply extend the Definition of SK query result and also return the paths connected $M$ with $v_i$. Then, it is well suited to the problem of finding ``possible relationships'' between two entities in RDF graph. Furthermore, it is also easy to extend our method for returning paths. During the backward search from vertices containing keywords, we only need specify the precursor vertex of each vertex.}

Given an SK query $\langle Q,q=\{w_1,...,w_n\}\rangle$, the cost of a result $r=\langle M, \{v_1, v_2, ... ,v_n\}\rangle$ contains two parts. The first part is the content cost and the second part is the structure cost.

\begin{definition} Given a result $r=\langle M, \{v_1, v_2, ... ,v_n\}\rangle$, the \emph{cost} of $r$ is defined as follows:
\[
Cost(r) = Cost_{content} (r) + Cost_{structure} (r)
\]
where $Cost_{content} (r)$ is the content cost of $r$ (defined in Definition \ref{def:answerscontent}) and $Cost_{structure}(r)$ is the structure cost of $r$ (defined in Definition \ref{def:structure}).
\end{definition}

\begin{definition}\label{def:answerscontent}
Given a result $r=\langle M, \{v_1, v_2, ... ,v_n\}\rangle$, the \emph{content cost} of $r=\langle M, \{v_1, v_2, ... ,v_n\}\rangle$ is defined as follows:
\[
Cost_{content} (r) = \sum\nolimits_{i = 1}^{i = n} C(v_i ,w_i )
\]
where $C(v_i,w_i)$ is the matching cost between $v_i$ and keyword $w_i$.
\end{definition}

Any typographic or linguistic distances, such as string edit distances \cite{DBLP:journals/jacm/WagnerF74}
and google similarity distance \cite{TKDE07:GoogleSimilarity}, can be used to measure $C(v_i,w_i)$.

In applications, users are more interested in some variables (in SPARQL query $Q$) than the constants in $Q$. Let us recall Example \ref{example:1}. The distance between keywords and the matching vertices with regard to variable ``?a'' is more interesting to measure the relationship strength. Therefore, to evaluate the structure cost (in Definition \ref{def:structure}), we only consider the matching vertices with regard to variables in SPARQL query $Q$.

\nop{
The distances between keywords and constants $\langle$Actor$\rangle$ are the same in the three matches $M_1$, $M_2$ and $M_3$, as shown in Figure \ref{fig:yagoexmapleofgraph}.
}

\begin{definition}\label{def:structure}

Given a result $\langle M, \{v_1, v_2, ... ,v_n\}\rangle$ for an SK query $\langle Q,q\rangle$, the distance between match $M$ and vertex $v_i$ ($i=1,...,n$) is defined as follows.
\[
d(M, v_i)=MIN_{v \in M}\{d(v,v_i)\}
\]

where $v$ is a matching vertex in $M$ with regard to a variable in SPARQL query $Q$ and $d(v, v_i)$ is the shortest path distance between $v$ and $v_i$ in RDF graph $G$.

Then, the \emph{structure cost} of a result $r=\langle M, \{v_1, v_2, ... ,v_n\}\rangle$ is defined as follows.
\[
Cost_{structure} (r) = \sum\nolimits_{i = 1}^{i = n} d (M,v_i )
\]
\end{definition}

\nop{
For simplicity of presentation, we only consider structure cost and ignore content cost in the following discussion until Section \ref{sec:onlineprocessing}. In Section \ref{sec:onlineprocessing}, we discuss how to handle both content and structure costs in our proposed algorithm.
}

(\textbf{Problem Definition}) \emph{Given an SK query $\langle Q,q\rangle$ and a parameter $k$, our problem is to find $k$ \emph{results} (Definition \ref{def:query}), which have the $k$-smallest costs.}

\subsection{Predicate Salience}\label{sec:predicatesalience}

In this paper, we use ``shortest path distance'' to evaluate the relation strength. However, the naive definition of the shortest path distance suffers from a critical problem: all predicates, i.e., edge labels, are considered equally important when it is used to measure the relationship strength between entities. In fact some predicates have little or no discriminating power in determining relevance. For example, predicates like ''type'' and ``label'' are so common that each entity is incident to a class vertex through an edge of predicate ``type''. This tends to incorrectly emphasize paths which contain these common predicates more frequently, without giving enough weight to the paths of more meaningful predicates (like ``actedIn'' and ``isMarriedTo''). The predicates like ''type'' and ``label'' are not good predicates to distinguish relevant and non-relevant vertices, unlike the less common predicates ``actedIn'' and ``isMarriedTo''.

Hence, we should introduce a mechanism for attenuating the effect of predicates that occur too frequently in the RDF graph to be meaningful for relevance determination. Learning from the concept of document frequency, we first find out the set of vertices in the RDF graph incident to a predicate $p$, which is denoted as $V(p)$. Then, we divide the size of $V(p)$ by the total number of vertices. We name the measure as the \emph{predicate salience} of predicate $p$ and give the formal definition as follows:
\[
ps(p) =  \frac{{|V(p)|}}{{|V(G)|}},
\]
Thus the predicate salience of a rare predicate is low, whereas the predicate salience of a frequent predicate is likely to be high, which means that rare predicates have less cost than frequent predicates.

Let us consider RDF graph in Figure \ref{fig:exampleRDFGraph}. The predicate salience values of all predicates are given in Table \ref{table:yagopredicateweights}.As shown in Table \ref{table:yagopredicateweights}, predicate ``actedIn'' is more important than ``type'' in measuring the relation strength, while the former's predicate salience is 0.296 and the latter's is 0.593.

\vspace{-0.1in}
\begin{table}[h]
\small
\centering
\begin{tabular}{|c|c|}
\hline
Predicate&Predicate Salience\\
  \hline
  \hline
  actedIn&$0.296$ \\
  \hline
  isMarriedTo&$0.074$\\
  \hline
  label&$0.852$\\
  \hline
  livesIn&$0.074$\\
  \hline
  type&$0.593$\\
  \hline
  wonPrize&$0.259$\\
  \hline
  \end{tabular}
  \vspace{-0.2in}
\caption{Weights of Predicates in Example RDF Graph}
\label{table:yagopredicateweights}
\end{table}
\vspace{-0.25in}

\section{Overview}\label{sec:overview}
In this section, we give an overview of the different steps involved in our process of SK query, which is depicted in Figure \ref{fig:SystemArchitecture}. In this paper, we are concerned with the challenge of efficiently finding the results of SK queries. We propose an approach in which the best results of the SK query are computed using the graph exploration. We detail the different steps of the approach below.

\vspace{-0.1in}
\begin{figure}[h]
\begin{center}
    \includegraphics[scale=0.27]{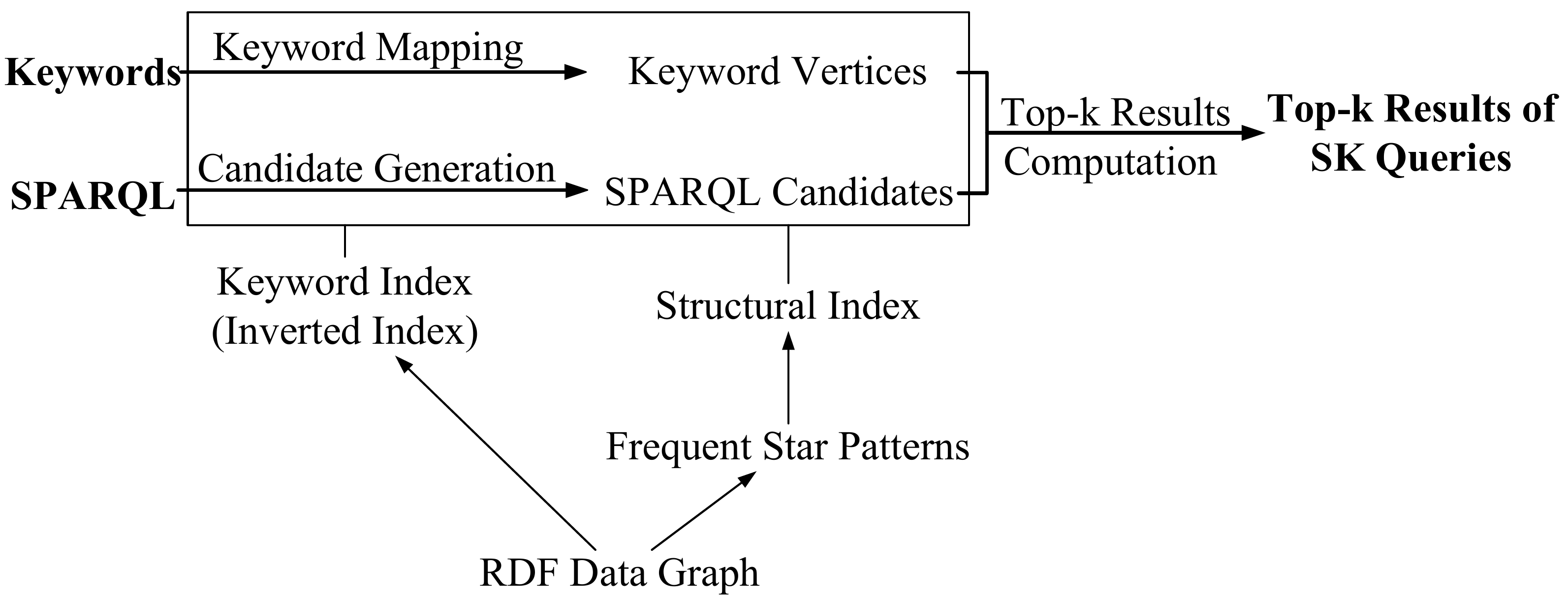}
    \vspace{-0.45in}
   \caption{Overview of Our Approach}
   \label{fig:SystemArchitecture}
\end{center}
\end{figure}
\vspace{-0.1in}

\textbf{Keyword Mapping.} In the offline phase, we create an inverted index storing a map from keywords to its locations in the RDF graph. In the online phase, we map keywords to vertices based on the inverted index.

For scoring keyword vertices, a widely used metric that is computed on-the-fly for a given query is IR-style TF/IDF cost. Many cost functions have been proposed in the literature, and we select one of them to assign the cost to each vertex containing keywords. Note that we need to normalize the cost of keyword matching vertices before the distances computation.

\nop{
Hence, for each keyword $w_i$, when we get a set $V_i$ of vertices containing $w_i$, we use the existing cost function to assign a vertex $v$ in $V_i$ with the matching cost. We denote the matching cost of between $v$ and $w_i$ as $C(w_i, v)$. When we start the search, for each vertex $v$ containing $w_i$, we initialize an element as $(v,\emptyset,C(v,w_i))$ and insert the element into a priority queue $Q_i$ for $w_i$. The priority queue $Q_i$ is sorted by $C(v,w_i)$. When a vertex $v$ is reachable to a keyword $w_i$, we assign a value of $d[v][i]$ to $v$, where $d[v][i]$ consists of not only the distance between $v$ and $w_i$ in RDF graph but also the content cost to record how well $w_i$ matches.
}

In this paper, our primary focus is indexing and query processing, so we will not delve into the specifics of keyword mapping.

\textbf{Candidate Generation.} When we find a vertex reachable to elements of all keywords, we need to run subgraph homomorphism to check whether there exist some subgraph matches (of $Q$) containing $v$. As we know, subgraph homomorphism is not efficient due to its high complexity \cite{DBLP:complexity}. In order to speed up query processing, we propose a filter-and-refine strategy to reduce the number of subgraph homomorphism operations. The basic idea is to filter out some vertices that are not in any subgraph match of $Q$. We call them ``\emph{dummy}'' vertices. If the search meets a dummy vertex, we do not perform subgraph homomorphism algorithm.

In this paper, we propose a frequent star pattern-based structural index. Based on this index, we can locate a candidate list in RDF graph of each variable in SPARQL. A vertex in at least one candidate lists of variables is not dummy. We will detail how to build the structural index in Section \ref{sec:structuralindex} and how to use the index to reduce the candidates of all variables in Section \ref{sec:CandidatesGeneration}.

\textbf{Top-k Results Computation.} Based on the keyword vertices and variables' candidates, we propose a solution based on graph exploration to compute out the top-k result of SK queries. Our approach starts graph exploration from all keyword vertices, and explores to their neighboring vertices recursively until the distances between a vertex and keyword vertices have been computed out. When the distances between a vertex and vertices of all keywords have been computed out, we check whether this vertex is a dummy vertex. If so, there exists no match of $Q$ that can contain it. Hence, we can skip it. Otherwise, we start our SPARQL matching algorithm (Algorithm \ref{alg:sparqlmatchingnew}) from the vertex to generate all matches containing it. The exploration terminates when the top-k results have been computed. We also propose some early stop strategies for top-k computation to reach early termination after obtaining the top-k results, instead of searching the data graph for all results.

We discuss the detail of top-k results computation in Section \ref{sec:topkcomputation}.

\vspace{-0.075in}
\section{Candidate Generation Based on Structural Index}\label{sec:index}
In this section, we first introduce an structural index based on a certain kind of patterns in Section \ref{sec:structuralindex}. Then, we discuss how to generate the candidate lists of variables based on our structural index \ref{sec:CandidatesGeneration}.

\vspace{-0.075in}
\subsection{Structural Index}\label{sec:structuralindex}
In this section, we propose a frequent star pattern-based index. We mine some frequent star patterns in $G$. For each frequent star $S$, we build an inverted list $L(S)$ that includes all vertices (in RDF graph $G$) contained by at least one match of $S$. A reason for selecting stars as index elements is that SPARQL queries tend to contain star-shaped subqueries, for combining several attribute-like properties of the same entity \cite{DBLP:rdf3x}.

We propose a sequential pattern mining-based method to find frequent star patterns in RDF graphs. For each entity vertex in an RDF graph, we sort all its adjacent edges in lexicographic order of edge labels (i.e. properties). These sorted edges can form a sequence. For example, vertex ``Philadelphia(film)'' has five adjacent edges, that are $\langle actedIn,$ $actedIn,$ $actedIn,$ $name,$ $type \rangle$. Table \ref{table:sequenceDB} shows a sequence database, where each sequence is formed by the adjacent edges of one entity vertex. We employ the existing sequential pattern mining algorithms, such as PrefixSpan\cite{DBLP:PrefixSpan}, to find frequent sequential patterns, where each sequential pattern corresponds to a star pattern in RDF graphs. For example, assume that the minimal support count $s=2$, $\langle actedIn,type \rangle$ and $\langle actedIn,type,wonPrize \rangle$ are two frequent sequential patterns. It is easy to know that a sequential pattern always corresponds to one star pattern, as shown in Figure \ref{fig:starpattern}. For ease of presentation, we use the terms ``sequential patterns'' and ``star patterns'' interchangeably in the following discussion.

\vspace{-0.1in}
\begin{table}[h]
\centering
\begin{tabular}{p{2.4cm}|p{5.8cm}}
\hline
Vertex&Predicate Sequence\\
  \hline
  \hline
  Philadelphia(film)&$\langle actedIn,actedIn,actedIn,label,type\rangle$\\
  \hline
  JoanneWoodward&$\langle actedIn,actedIn,label,type,wonPrize\rangle$\\
  \hline
  AntonioBanderas&$\langle actedIn,label,type,wonPrize,wonPrize\rangle$\\
  \hline
  ...&...\\
  \hline
  \end{tabular}
  \vspace{-0.2in}
\caption{Example of Predicate Sequences}
\label{table:sequenceDB}
\end{table}
\vspace{-0.25in}

\begin{figure}[h]
\begin{center}
    \includegraphics[scale=0.34]{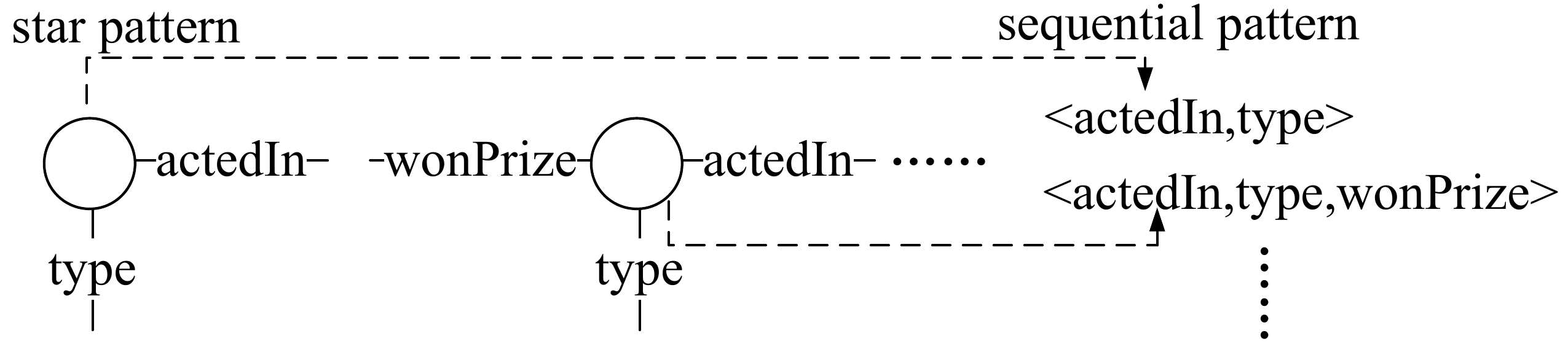}
    \vspace{-0.5in}
   \caption{ Example Star Pattern}
   \label{fig:starpattern}
\end{center}
\end{figure}
\vspace{-0.1in}

For each frequent star pattern $S$, we maintain an inverted list $L(S)=\{v |S$ occurs in $v$'s adjacent edge sequence$\}$. Obviously, if we use all frequent stars as the index elements, the space cost is very large. Thus, inspired by gIndex \cite{DBLP:gIndex}, we also define the \emph{discriminative ratio} for the star pattern selection.
\begin{definition}\label{def:selectivepattern}
Given a star $S$, its discriminative ratio is defined as follows:
\[
\gamma (S) = \frac{{|L(S)|}}{{| \cap _{S^\prime  \subset S} L(S^\prime )|}}
\]
where $S^\prime \subset S$ denotes that $S^\prime$ is a part of $S$.
\end{definition}

Obviously, $\gamma(S) \le 1$. $\gamma(S)=1$ means that $L(S)$ can be obtained by the intersection of all $L(S^\prime)$, where $S^\prime \subset S$. In this case, if all $S^\prime$ are index elements, it is not necessary to keep $S$ as the index element, as $S$ cannot provide more pruning power. In practice, we set a threshold $\gamma_{max}$, and we only choose the stars $S$, where $\gamma(S) \le \gamma_{max}$. Note that, to ensure the completeness of the indexing, we always
choose the (absolute) support to be 1 for size-1 stars (star with only one edge). This method can guarantee no-false-negative, since all vertices (in $G$) are indexed in at least one inverted list.

\vspace{-0.075in}
\begin{theorem}\label{theorem:starpruning} Let $F$ denote all selected index elements (i.e, frequent star patterns). Given a SPARQL query $Q$, a vertex $v$ in graph $G$ can be pruned (there exists no subgraph match of $Q$ containing $v$) if the following equation holds.
\begin{equation}
v \notin \bigcup\nolimits_{S \in F \wedge S \in Q} {L(S)}
\end{equation}
where $S \in F$ means that $S$ is a selected star pattern and $S \in Q$ is a star pattern included in $Q$.
\end{theorem}
\begin{proof} If $v \notin \bigcup\nolimits_{S \in F \wedge S \in Q} {L(S)}$, it means that the structure around $v$ does not contain any substructure of $Q$. Hence, $v$ must be unable to in a match of $Q$.
\end{proof}

\subsection{Candidate Generation}\label{sec:CandidatesGeneration}
Given an SK query, we first tag the vertices that can be pruned by Theorem \ref{theorem:starpruning}. For each variable in SPARQL, we locate its candidates in RDF graph. Each variable can map to a predicate sequence according to the SPARQL statement. For example, variable ``?a'' of the SPARQL query in Figure \ref{fig:ExampleSPARQLKeyword} has the predicate sequence $\langle actedIn, type \rangle$. Then, for each variable $x$, we look up our structural index and find the maximum pattern contained by $x$'s predicate sequence. We load the vertex list of the maximum pattern as $x$'s candidates. A vertex in at least one vertex lists of variables is not dummy. We define these pruned vertices as \emph{dummy vertices} as follows.

\vspace{-0.1in}
\begin{definition}\label{def:dummyvertex}\textbf{Dummy Vertex.} Given a SPARQL query $Q$, a vertex $v$ in graph $G$ is called as \emph{dummy vertex} if the following equation holds.
\begin{equation}
v \notin \bigcup\nolimits_{S \in F \wedge S \in Q} {L(S)}
\end{equation}
where $F$ denote all selected frequent star patterns, $S \in F$ means that $S$ is a selected star pattern and $S \in Q$ is a star pattern included in $Q$.
\end{definition}

\vspace{-0.1in}

When the search process meets a fully-seen vertex $v$, if $v$ is not a dummy vertex, we perform subgraph isomorphism algorithm to find the subgraph match of SPARQL query $Q$ containing $v$. Otherwise, we do not perform subgraph isomorphism algorithm beginning with $v$.

\vspace{-0.1in}
\section{Top-k Results Computation}
\label{sec:topkcomputation}

In this section, we introduce our approach for SK queries, which is based on the ``backward search'' strategy \cite{DBLP:BANKS}. Our algorithm for searching top-k results of SK queries is shown in Algorithm \ref{alg:basicbackwardsearch}. This algorithm consists of three parts: 1) graph exploration to find vertices connecting the keyword vertices, 2) generation of SPARQL matches from the vertices connecting the keyword vertices and 3) top-k computation. In the following, we will elaborate on these three tasks.

\subsection{Graph Exploration}\label{sec:backwardsearch}
Given the keyword vertices, the objective of the exploration is to find vertices in the graph that connect these keyword vertices and compute their distances to these keyword vertices. Let $V_i$ denote all literal vertices (in RDF graph $G$) containing keyword $w_i$.

\vspace{-0.1in}
\begin{definition} \textbf{Distance between Vertex and Keyword}. Given a vertex $v$ in RDF graph $G$ and a keyword $w_i$, the distance between $v$ and keyword $w_i$ (denoted as $d(v,w_i)$) is the minimum distance between $v$ and a vertex in $V_i$, where $V_i$ includes all literal vertices containing keyword $w_i$ in $G$.
\end{definition}
\vspace{-0.1in}

\begin{algorithm}[t] \label{alg:basicbackwardsearch}\small
\caption{Search for Top-k Results of SK Queries}

\KwIn{RDF data graph $G$, SK query $\langle Q,q\rangle$, $\{V_1,...,V_n\}$ where $V_i$ is the set of vertices containing keyword $w_i$, priority queues $\{PQ_1,...,PQ_n\}$.}
\KwOut{ Top-k results $R$ of $\langle Q,q\rangle$.}

\For{each vertices set $V_i$}
{
    \For{each vertex $v$ in $V_i$}
    {
        Insert $(v,\emptyset,0)$ into $PQ_i$;
    }
}
\While{not all queues are empty}{
    \For{$i=1,...,n$}
    {
        Pop the head of $PQ_i$ (denoted as $(v,p,|p_i|)$), set $d[v][i]=|p_i|$ and insert it into $RS_i$;\\
        \For{each adjacent edge $\overline{vv^\prime}$ to $v$}
        {
            \If{$p \cup \overline{vv^\prime}$ is not a simple path}
            {
                Continue;
            }
            \If{there exists another element $(v^\prime, p^\prime, |p^\prime|)$ in $PQ_i$}
            {
                \If{$|p^\prime| > |p| + ps(\overline{vv^\prime}))$}
                {
                    Delete $(v^\prime,p^\prime,|p^\prime|)$;\\
                    Insert $(v^\prime,p \cup \overline{vv^\prime}, |p_i|+ps(\overline{vv^\prime}))$ in $PQ_i$;
                }
                \Else
                {
                    Continue;
                }
            }
            \Else
            {
                Insert $(v^\prime,p \cup \overline{vv^\prime}, |p|+ps(\overline{vv^\prime}))$ in $PQ_i$;
            }
        }
        \If{$v$ is a fully-seen vertex}
        {
            Call Algorithm \ref{alg:sparqlmatchingnew} to find all matches containing $v$;
        }
        \For{each match $M$ containing vertex $v$}
        {
            \If{all vertices in $M$ are fully-seen vertices}
            {
                Use $M$ to update $R$ and the upper bound $\delta$ of top-k results
            }
        }
        Update the cost of all partially-seen matches and $\delta$;\\
        Update the lower bound cost $\theta$ of all remaining un-seen vertices;\\
    }
    \If{$\theta\ge\delta$}
    {
        Break;
    }
}
Return $R$.
\end{algorithm}

For graph exploration, we maintain a priority queue $PQ_i$ for each keyword $w_i$. Each element in $PQ_i$ is represented as $(v,p,|p|)$, where $v$ is a vertex id, $p$ is a path between $v$ and a vertex in $V_i$ and $|p|$ denotes the path distance. All elements in $PQ_i$ are sorted in the non-descending order of $|p|$. Each keyword $w_i$ is also associated with a result set $RS_i$. In order to keep track of information related to each vertex $v$, we associate $v$ with a vector $d[v]$.  If a vertex $v$ is in $RS_i$ ($i=1,...,n$), the shortest path distance is known. In this case, we set $d[v][i]=d(v,w_i)$; otherwise, we set $d[v][i]=null$.

Initially, the exploration starts with a set of vertices containing keywords. For each vertex $v$ containing keyword $w_i$, an element $(v,\emptyset,0)$ is created and placed into the queue $PQ_i$ (Line 3 in Algorithm \ref{alg:basicbackwardsearch}). During the search, at each step, we pick a queue $PQ_i$ ($i=1,...,n$) to expand in a round-robin manner (Line 5 in Algorithm \ref{alg:basicbackwardsearch}). We assume that we pop the queue head $(v,p,|p|)$ from $PQ_i$. When a queue head $(v,p,|p|)$ is popped from queue $PQ_i$, we insert it into result set $RS_i$ and set ${d}[v][w_i]$ = $|p|$ (Line 6 in Algorithm \ref{alg:basicbackwardsearch}). We can prove that the following theorem holds.

\begin{theorem}\label{theorem:distancecorrectness} When a queue head $(v,p,|p|)$ is popped from queue $PQ_i$, the following equation holds.
\[
d(v,w_i )=d[v][i]=|p|
\]
\end{theorem}

\begin{proof}
Given a vertex $v$, before $(v,p,|p|)$ is popped from $PQ_i$, a path $p$ between $v$ and vertices containing $w_i$. It is obvious that $|p|\ge d(v,w_i )$.

We wish to show that in each iteration, $d(v,w_i)=d[v][i] = |p|$ for the element $(v,p,|p|)$ popped from $PQ_i$. We prove this by contradiction. We assume that $v$ is the first vertex for which $d[v][i]=|p| \ne d(w_i, v)$ when $(v,p,|p|)$ is popped from $PQ_i$. We focus our attention on the situation at the beginning of the iteration in which $(v,p,|p|)$ is popped from $PQ_i$ and derive the contradiction that $d[v][i]=|p|= d(v,w_i)$ at that time by examining the shortest path from $v$ to vertices containing $w_i$. We must have $v\notin V_i$ because all vertices in $V_i$ are the first vertices added to set $RS_i$ and $d[v][i] = 0$ at that time.

Because $v\notin V_i$, we also have that $RS_i\notin \emptyset$ just before $(v,p,|p|)$ is popped from $PQ_i$. There must be some paths from vertices containing $w_i$ to $v$, for otherwise $d[v][i] = \infty$ by the no-path property, which would violate our assumption that $d[v][i] \ne d(w_i, v)$. Because there is at least one path, there is the shortest path $p^\prime$ between $v$ and vertices in $V_i$. Prior to pop $(v,p,|p|)$ to $PQ_i$, path $p^\prime$ connects a vertex in $RS_i$, namely some vertices in $V_i$, to a vertex in $V(G) - RS_i$, namely $v$. Let us consider the first vertex $v^\prime$ along $p^\prime$ such that $v^\prime\in V(G) - RS_i$, and let $v^{\prime\prime}\in RS_i$ be the predecessor of $v^\prime$.

We claim that $d[v^{\prime}][i] = d(w_i, v^{\prime})$ when the element of $v^{\prime}$ is popped from $PQ_i$. To prove this claim, observe that $v^{\prime\prime}\in RS_i$. Then, because $v$ is chosen as the first vertex for which $d[v][i] \ne d(w_i, v)$ when $(v,p,|p|)$ is popped from $PQ_i$, we have $d[v^{\prime}][i] = d(w_i, v^{\prime})$ when $v^{\prime}$ is added to $RS_i$. Edge $\overline{v^{\prime}v^{\prime\prime}}$ is relaxed at that time (Line 7 - 17 in Algorithm \ref{alg:basicbackwardsearch}), so the claim follows from the convergence property.

We can now obtain a contradiction to prove that $d[v][i] = d(v, w_i)$. Because $v^{\prime}$ occurs before $v$ on the shortest path from vertices in $V_i$ to $v$ and all edge weights are nonnegative, we have $d[v^{\prime}][i]\le d(v,w_i)$, and thus $d(v^{\prime},w_i)=d[v^{\prime}][i]\le  d(v,w_i) \le d[v][i]$.

But because both vertices $v$ and $v^{\prime}$ are in $V(G) - RS_i$ when $v^{\prime}$ is popped before $v$, we have $d(v^{\prime},w_i) \le d(v,w_i)$. Thus, $d(v^{\prime},w_i)=d[v^{\prime}][i]=d(v,w_i)$, which contradicts our choice of $v$. We conclude that $d[v][i] = d(w_i, v)$ when $(v,p,|p|)$ is popped from $PQ_i$, and that this equality is maintained at all times thereafter.
\end{proof}

When a queue head $(v,p,|p|)$ is popped from queue $PQ_i$, it means that we have computed out the distance between $v$ and keyword $w_i$. We also says that $v$ is \emph{seen} by keyword $w_i$.

\begin{definition} \textbf{Seen by Keyword}. When a queue head $(v,p,|p|)$ is popped from queue $PQ_i$, we say vertex $v$ is \emph{seen} by keyword $w_i$.
\end{definition}

Assume that $(v,p,|p|)$ is popped from queue $PQ_i$. For each incident edge $\overline{vv^\prime}$ to $v$, we obtain a new element $(v^{\prime},p \cup \overline{vv^{\prime}}, |p|+ps(\overline{vv^{\prime}}))$, where $p \cup \overline{vv^{\prime}}$ means appending an edge to $p$ and $ps(\overline{vv^{\prime}})$ denotes the predicate salience value of the edge label of $\overline{vv^{\prime}}$. If $p_i \cup \overline{vv^{\prime}}$ is not a simple path\footnote{A simple path is a path with no repeated vertices.}, the element is ignored (Line 8-9 in Algorithm \ref{alg:basicbackwardsearch}). Then, we check whether there exists another element $(v^{\prime},p^\prime,|p^\prime|)$ that has the identical vertex $v^{\prime}$ with the new element $(v^{\prime},p \cup \overline{vv^{\prime}}, |p|+ps(\overline{vv^{\prime}}))$, where $|p^\prime|>|p|+ps(\overline{vv^{\prime}}))$. If so, we delete $(v^{\prime},p^\prime,|p^\prime|)$ from $PQ_i$ and insert $(v^{\prime},p \cup \overline{vv^{\prime}}, |p|+ps(\overline{vv^{\prime}}))$ into $PQ_i$ (Line 12-13 in Algorithm \ref{alg:basicbackwardsearch}). Otherwise, we ignore the new element (Line 15 in Algorithm \ref{alg:basicbackwardsearch}). If there exists no element $(v,p^\prime,|p^\prime|)$, we insert $(v^{\prime},p \cup \overline{vv^{\prime}}, |p|+ps(\overline{vv^{\prime}}))$ into the queue directly (Line 17 in Algorithm \ref{alg:basicbackwardsearch}).

\begin{definition} \textbf{Fully-seen Vertex, Partially-seen Vertex and Un-seen Vertex}. Given a vertex $v$, if $v$ is seen by all keywords $w_i$ ($i=1,...,n$), $v$ is called a \emph{fully-seen} vertex; if $v$ is not a fully-seen vertex but it has been seen by at least one keyword, $v$ is called a \emph{partially-seen} vertex; if $v$ is not seen by any keyword, $v$ is called an \emph{un-seen} vertex.
\end{definition}
\vspace{-0.05in}

At each step, we check whether the vertex just popped from the queue has seen by all keywords. Specifically, for a popped queue head $v$, if all dimensions of its vector $d[v]$ are non-null, it means that all keywords have seen vertex $v$, i.e., we have known the distance between $v$ and each keyword. In this case, $v$ is a \emph{fully-seen} vertex. When we meet a fully-seen vertex $v$, we will employ a subgraph homomorphism algorithm to find matches containing $v$ (Line 18-19 in Algorithm \ref{alg:basicbackwardsearch}). The details will be discussed in Section \ref{sec:computingsparqlmatch}.

\vspace{-0.05in}
\subsection{Generation of SPARQL Matches}\label{sec:computingsparqlmatch}
When we find out a fully-seen vertex $v$, it means that we have known the distance between $v$ and each keyword $w_i$. The next step is to compute SPARQL matches containing vertex $v$ if any. Here, we perform subgraph homomorphism algorithm to find subgraph matches (of query $Q$) containing $v$.

Generally speaking, we employ a DFS-based \emph{state} \emph{transformation} algorithm to perform the matching process beginning from a fully-seen vertex $v$ (as shown in Algorithm \ref{alg:sparqlmatchingnew}). Here, we define the \emph{state} as follows.

\begin{definition} Given a SPARQL query graph $Q$ with $m$ vertices $u_1,...,u_m$, a \emph{state} is a (partial) match of query graph $Q$.
\end{definition}

For example, Figure \ref{fig:partialmatch} shows an example state of the SPARQL query in Figure \ref{fig:ExampleSPARQLKeyword}.

\vspace{-0.1in}
\begin{figure}[h]
\begin{center}
    \includegraphics[scale=0.25]{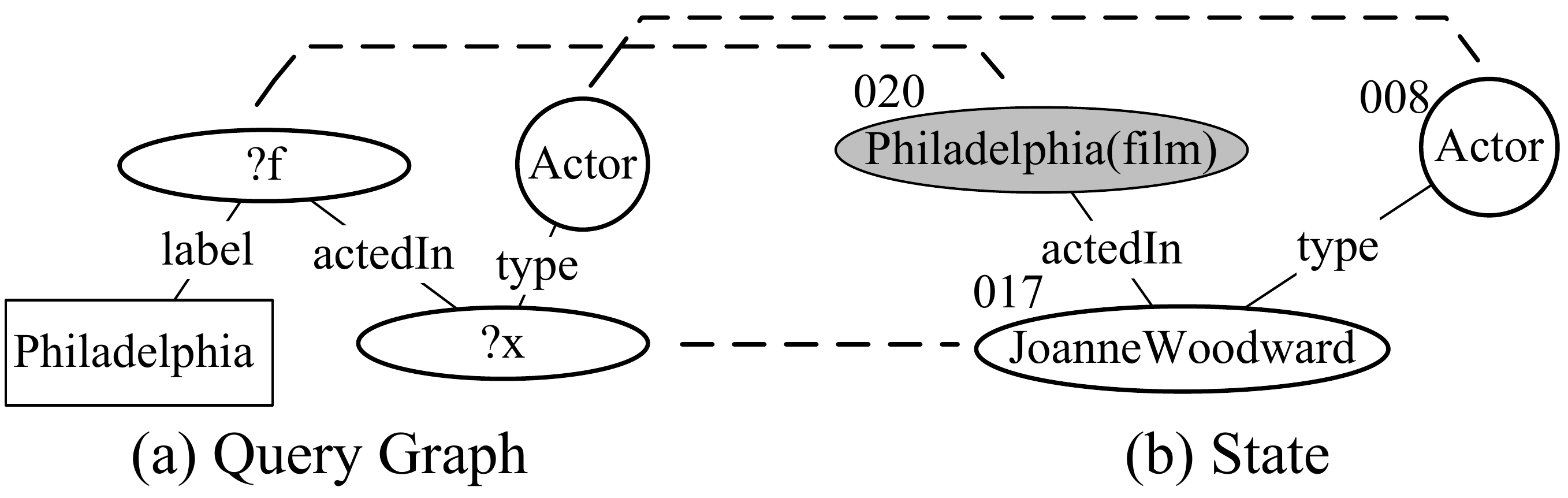}
    \vspace{-0.3in}
   \caption{ A State Matching a Part of $Q$}
   \label{fig:partialmatch}
\end{center}
\end{figure}
\vspace{-0.1in}

\begin{figure*}
\begin{center}
    \includegraphics[scale=0.24]{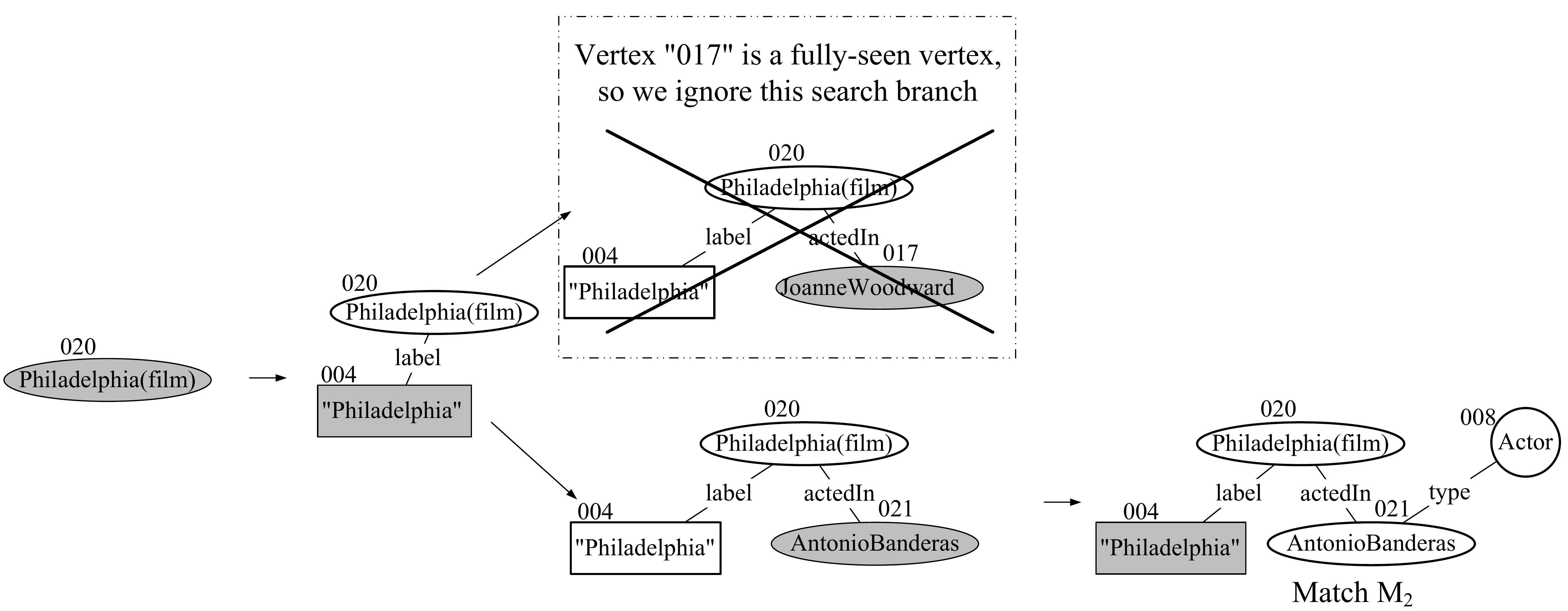}
   \caption{ Finding Matches Containing Vertex ``$020$'' by Pruning the Search Branch Beginning from ``$017$''}
   \label{fig:statepruning}
   \vspace{-0.35in}
\end{center}
\end{figure*}

In particular, our \emph{state} \emph{transformation} algorithm is as follows. Assume that $v$ matches vertex $u$ in SPARQL query $Q$. We first initialize a state with $v$. Then, we search the RDF data graph to reach $v$'s neighbor $v^\prime$ corresponding to $u^\prime$ in $Q$, where $u^\prime$ is one of $u$'s neighbors and edge $\overline{vv^\prime}$ satisfies query edge $\overline{uu^\prime}$. The search will extend the state step by step. The search branch terminates until that we have found a state corresponding to a match or we cannot continue. In this case, the algorithm is backtrack to some other states and try other search branches.

\begin{algorithm}[t] \label{alg:sparqlmatchingnew}\small
\KwIn{ A candidate vertex $v$ corresponding to $u$ in SPARQL query $Q$, and a state stack $S$.}
\KwOut{ The match set $MS$ of $Q$ containing $v$.}

Initialize a state $s$ with $v$;\\
Push $s$ into $S$;\\
\While{$S \neq \emptyset$}
{
    Pop the first state $s\in S$;\\
    \If{all edges of $Q$ have been matched in $s$}
    {
        Insert $s$ to $MS$;
    }
    \For{each unmatched edge $\overline{u^\prime u^{\prime\prime}}$ that $u^\prime$ has been matched to $v^\prime$}
    {
        \If{$u^{\prime\prime}$ has been matched to $v^{\prime\prime}$}
        {
            \If{$\overline{v^{\prime} v^{\prime\prime}}\in E(G)$}
            {
                Push $s$ into $S$;
            }
            \Else
            {
                Continue;
            }
        }
        \Else{
            \For{each neighbor $v^{\prime\prime}$ of $v^{\prime}$}
            {
                \If{$v^{\prime\prime}$ is a dummy or fully-seen vertex}
                {
                    Continue;
                }
                \If{$\overline{v^{\prime} v^{\prime\prime}}$ can match $\overline{u^\prime u^{\prime\prime}}$}
                {
                    Initialize a new state $s^\prime$ and $s^\prime= s$;\\
                    Match $u^{\prime\prime}$ with $v^{\prime\prime}$;\\
                    Push $s$ into $S$;
                }
            }
        }
    }
}
Return $MS$.
\caption{SPARQL Matching Algorithm}
\end{algorithm}

As shown in Algorithm \ref{alg:sparqlmatchingnew}, we find all matches containing some fully-seen vertex $v$ only if $v$ is not be a dummy vertex(Line 15 - 16 in Algorithm \ref{alg:sparqlmatchingnew}). This is because that there exists no subgraph match containing a dummy vertex. When we finish Algorithm \ref{alg:sparqlmatchingnew} from $v$, we say that $v$ is \emph{searched}. The ``searched'' indicates that all matches containing $v$ has been found if any. When we search the RDF graph beginning with a fully-seen vertex, if the search meets another fully-seen vertex $v^{\prime\prime}$, it can skip $v^{\prime\prime}$ (Line 15 - 16 in Algorithm \ref{alg:sparqlmatchingnew}). This is because that the matches containing $v^{\prime\prime}$ have been found before.

\vspace{-0.05in}
\begin{example}
We assume that the current popped fully-seen vertex is ``$020$'' and vertex ``$017$'' is another a fully-seen vertex. As shown in Figure \ref{fig:statepruning}, we explore the RDF graph from ``$020$'' to ``$017$''. However, vertex ``$017$'' is a fully-seen vertex, so all SPARQL matches containing ``$017$'' have been found already. As a result, we can terminate the corresponding search branches in Figure \ref{fig:statepruning}.
\end{example}

\vspace{-0.1in}
\subsection{Top-k Computation}\label{sec:topk}
The native solution for computing top-k results of a SK query is to run the backward search algorithm until that all vertices (in RDF graph $G$) have been fully-seen by keywords. Then, according to the results' cost, we can find the top-k results. Obviously, this is an inefficient solution especially when $G$ is very large. In this subsection, we design an early-stop strategy.

Let us consider a snapshot of some iteration step in Algorithm \ref{alg:basicbackwardsearch}. All subgraph matches of SPARQL query $Q$ can be divided into three categories: fully-seen matches, partially-seen matches and un-seen matches.

\vspace{-0.05in}
\begin{definition}\textbf{Fully-seen Match, Partially-seen Match and Un-seen Match.} Given a subgraph match $M$ of SPARQL query $Q$, if all vertices in $M$ are fully-seen vertices, $M$ is called a \emph{fully-seen match}; if $M$ is not a fully-seen match and $M$ contains at least one fully-seen vertex, it is called  a \emph{partially-seen match}. If a match $M$ does not contain any fully-seen vertex, it is called an \emph{un-seen match}.
\end{definition}

Figure \ref{fig:visualspt} demonstrates a visual representation of three kinds of matches. The shaded area covered by the dash line circle denotes all fully-seen vertices in RDF graph. With the increasing of the iteration steps (in Algorithm \ref{alg:basicbackwardsearch}), the shaded area expands gradually until that it covers the whole RDF graph. The early-stop strategy is to stop the expansion as early as possible, but we can guarantee that we have found the top-k results for SK queries.

\begin{figure}[h]
\begin{center}
    \includegraphics[scale=0.3]{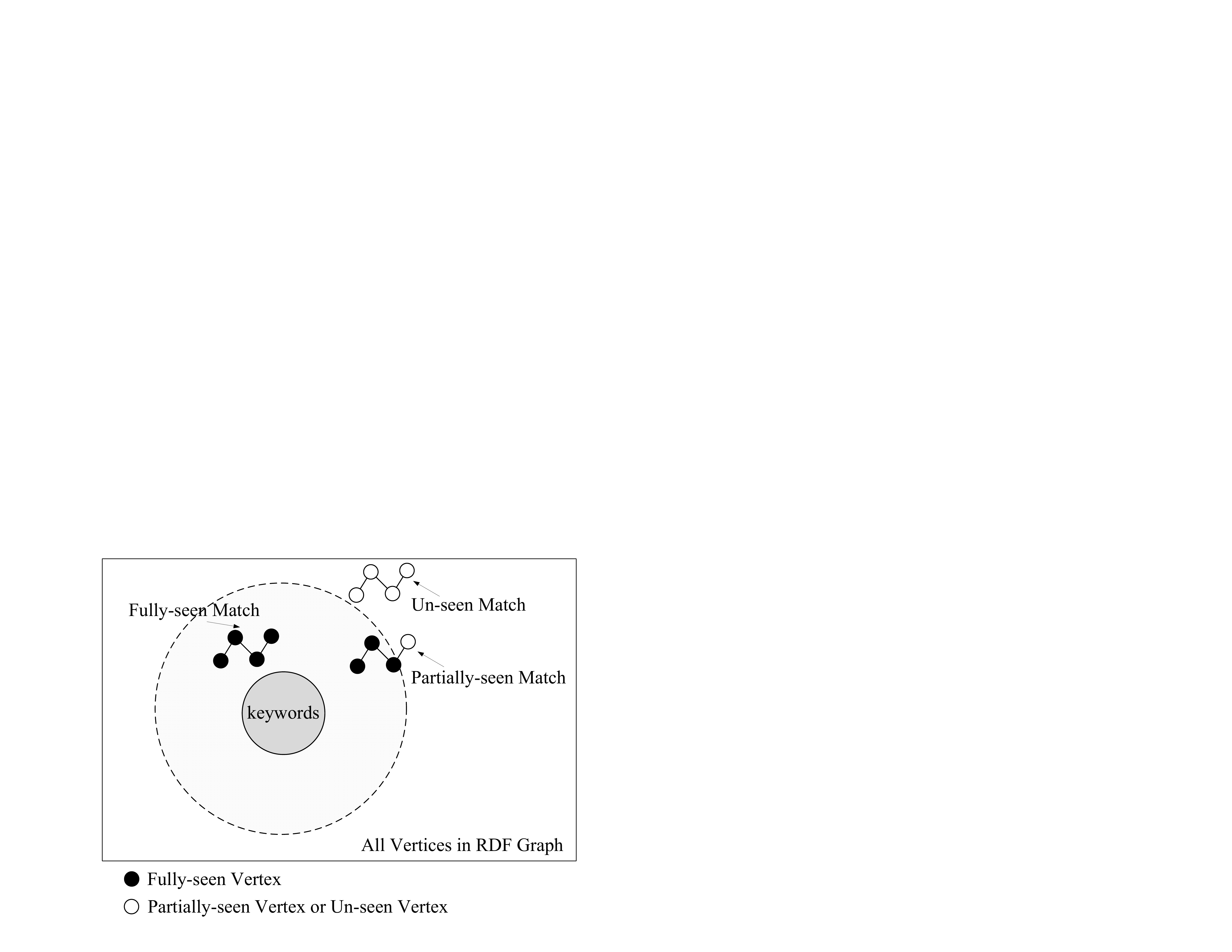}
   \caption{ Fully-seen Match, Partially-seen Match and Un-seen Match during the Backward Search}
   \vspace{-0.3in}
   \label{fig:visualspt}
\end{center}
\end{figure}

\nop{
In Algorithm \ref{alg:basicbackwardsearch}, only when we meet a fully-seen vertex, we find all subgraph matches containing the vertex. It also means that we always find all fully-seen matches and partially-seen matches but we do not identify un-seen matches. Let $FS$, $PS$ and $NS$ denote all fully-seen matches,  partially-seen matches and un-seen matches, respectively.
}

The basic idea of our early-stop strategy is as follows. We only compute the cost of fully-seen matches. Then, we use the fully-seen matches to find a threshold $\delta$, which is the $k$-th smallest cost so far. If there are less than $k$ fully-seen matches so far, $\delta$ is $\infty$. We compute the lower bounds $\theta_1$ and $\theta_2$ for partially-seen matches and un-seen matches, respectively. The algorithm can early stop if and only if $\delta <\theta_1 \wedge \delta <\theta_2$. Otherwise, the algorithm continues the next iteration.

\textbf{Fully-seen Match.} For a fully-seen match, we compute its match cost according to Definition \ref{def:structure}. If we have found more than $k$ fully-seen matches, we maintain a threshold $\delta$ be the $k$-th smallest match cost.

\textbf{Partially-seen Match.} For any partially-seen match, we compute the lower bound of its cost as follows.

\vspace{-0.05in}
\begin{theorem}\label{theorem:matchlowerbound} Given a partially-seen match $M$ of SPARQL query $Q$, $v$ is a partially-seen vertex or an un-seen vertex  in the match. The following equation holds.
\[
\begin{array}{l}
 Cost(M) =  \sum\limits_{1 \le i \le n} {d(v,w_i )}   \\
 \ge \sum\limits_{d[v][w_i ] \ne null \wedge 1 \le i \le n} {d[v][w_i]}  + \sum\limits_{d[v][w_i ] = null \wedge 1 \le i \le n} {|p_i |}  \\
 \end{array}
\]

where $d[v][w_i]$ is the $i$-th dimension of $v$'s vector corresponding to keyword $w_i$, and $|p_i|$ corresponds to the current queue head $(v,p_i,|p_i|)$ in queue $PQ_i$.
\end{theorem}
\begin{proof} If $d[v][i ] \ne null$, it means that we have computed out $d(v,w_i )$. If $d[v][w_i] = null$, $v$ has still not been seen. Since each time we pop the head $(v,p_i,|p_i|)$ of $PQ_i$ where $|p_i|$ is the smallest, all un-seen vertices' distances to $w_i$ are larger than $|p_i|$.
\end{proof}

According to Theorem \ref{theorem:matchlowerbound}, we define the lower bound of a partially-seen match $M$ as follows.

\begin{definition}\label{def:lowerbound} Given a match of SPARQL query $Q$, the lower bound for a partially-seen match $M$ is defined as follows.
\[
lb(M) = \mathop {MIN}\limits_{v \in M} (\sum\limits_{d[v][w_i ] \ne null \wedge 1 \le i \le n} {d[v][w_i ]}  + \sum\limits_{d[v][w_i ] = null \wedge 1 \le i \le n} {|p_i |} )
\]

\end{definition}

The lower bound for all partially-seen matches is defined as follows.

\begin{definition} The lower bound $\theta_1$ for all partially-seen matches is as follows.
\[
\theta _1  = MIN_{M \in PS} (lb(M))
\]
where $PS$ denotes all partially-seen matches and $lb(M)$ is defined in Definition \ref{def:lowerbound}.
\end{definition}

With the increasing of the iteration steps, some partially-seen matches become fully-seen matches. They are moved to $FS$. The threshold $\delta$ and $\theta_1$ are updated accordingly.

\textbf{Un-seen Match.} Let us consider an un-seen match $M$. There are two kinds of vertices in $M$, i.e., partially-seen vertices and un-seen vertices.

\vspace{-0.05in}
\begin{theorem}\label{theorem:2} For an un-seen vertex $v$, if threshold $\delta$ $\neq$ $\infty$, the following equation holds.
\[
\delta  \le \sum\limits_{1 \le i \le n} {d(v,w_i )}
\]

\end{theorem}
\begin{proof}
For each keyword $w_i$, we assume that the queue head of $PQ_i$ is $(v,p_i,|p_i|)$. Since $v$ is an un-seen vertex, $|p_i|\le d(v,w_i)$ for each keyword $w_i$. In contrast, $\delta$ is the upper bound of the top-k results, so $\delta$ is equal to the cost of a fully-seen match $M$. Each vertex $v^\prime$ in $M$ is fully-seen vertex. Hence, $d(v^\prime,w_i)\le |p_i|$. Then, we know that $d(v^\prime,w_i)\le |p_i|\le d(v,w_i)$ for each keyword $w_i$. In conclusion, $\delta  \le \sum\limits_{1 \le i \le n} {d(v,w_i )}$.
\end{proof}

According to Theorem \ref{theorem:2}, it is not necessary to consider un-seen vertices to define the lower bound for un-seen matches. Therefore, we define the lower bound for all un-seen matches as follows.

\begin{definition}\label{def:lowerboundunseenvertex} The lower bound $\theta_2$ for all un-seen matches is as follows.
\[
\theta_2  = \mathop {MIN}\limits_{v \in PSet} (\sum\limits_{d[v][w_i] \ne null \wedge 1 \le i \le n} {d[v][w_i ]}  + \sum\limits_{d[v][w_i] = null \wedge 1 \le i \le n} {|p_i |} )
\]
where $PSet$ contains all partially-seen vertices so far, $d[v][w_i]$ is the $i$-th dimension of the $v$'s vector corresponding to keyword $w_i$ and $|p_i|$ corresponds to the current queue head $(v,p_i,|p_i|)$ in queue $PQ_i$.
\end{definition}

\textbf{Early-stop Strategy}. In each iteration step, we check whether $\delta \le \theta_1 \wedge \delta \le \theta_2$. If the condition holds, the algorithm can stop, since any  partially-seen match or un-seen match cannot be in one of the top-k results.
\vspace{-0.1in}

\section{Experiments}\label{sec:experiments}
In this section, we evaluate our approach in three large real RDF graphs, DBLP, Yago and DBPedia.

For effectiveness study, we compare our method with a classical keyword search algorithm BANKS \cite{DBLP:BANKS} over both Yago and DBPedia. Furthermore, since each resource in DBPedia is annotated by Wikipedia documents, so we design a stronger baseline named as ``Annotated SPARQL'' for DBPedia. ``Annotated SPARQL'' is similar to the approach discussed in \cite{DBLP:CE2}. It first finds out all matches of the SPARQL query, then ranks these matches by how closely the corresponding Wikipedia documents match the keywords. Note that, except for DBpedia, most current RDF datasets do not provide such documents to annotate the resources. Hence, we only do experiments of annotated SPARQL over DBPedia. For other RDF datasets, although we can crawl some pages to annotate their entities, that is beyond the scope of this paper.

For efficiency study, because there are no existing method for SK queries, we evaluate our approach with two baselines, i.e., the \emph{exhaustive computing} and the \emph{naive backward search} and . ``Exhaustive computing'' has been introduced in Section \ref{sec:introduction}. The \emph{naive backward search} is to run the backward search algorithm until that all vertices have been fully-seen by keywords. Then, according to the results' cost, we find the top-k results.

Our experiments are conduct on a
machine with 2 Ghz Core 2 Duo processor, 16G RAM memory and running
Windows Server 2008. All experiments are implemented in Java languages. We use MySQL to store the RDF graphs and the indices.

\subsection{Datasets \& Setup}\label{sec:dataset}
We use three real-world RDF datasets, DBLP, Yago and DBPedia in our experiments. The details about the two datasets are as follows.

\textbf{DBLP}. DBLP \footnote{http://sw.de ri.org/~aharth/2004/07/dblp/} contains bibliographic information of computer science publications \cite{DBLP:dblp}. The DBLP graph contains $8,381,858$ RDF triples and $3,103,614$ vertices. We define $5$ sample SK queries for DBLP and show two of them in Table \ref{table:samplequerydblp} for case study.

\textbf{Yago}. Yago\footnote{http: //www.mpi-inf.mpg.de/yago-naga/yago/} extracts facts from Wikipedia and integrates them with the WordNet
thesaurus \cite{DBLP:yago}. The RDF graph has $19,012,849$ edges and $12,811,222$ vertices. We define $8$ sample SK queries for Yago and show two of them for case study in Table \ref{table:samplequeryyago}.

\textbf{DBPedia $\&$ QALD}. DBPedia \footnote{http://downloads.dbpedia.org/3.7/en/} is an RDF dataset extracted from Wikipedia. The DBPedia contains $73,766,900$ edges and $13,100,739$ vertices. QALD \footnote{http://greententacle.techfak.uni-bielefeld.de/$\sim$cunger/qald/index.php?x =challenge$\&$q=2} is an evaluation campaign on question answering over linked data. It is co-located with the ESWC 2012. In this campaign, the committee provides some questions and each question is annotated with some recommended keywords and the answers that these queries retrieve. Note that, some questions in QALD are so simple that they can map to a SPARQL query with only one edge. These simple questions are unnecessary to be split into a SPARQL query and some keywords. Thus, we only select $10$ non-aggregation complex queries from QALD for evaluation. Two of them are as shown in Table \ref{table:samplequerydbpedia} for case study.

All sample queries are shown in Appendix.

\subsection{Effectiveness Study}\label{sec:effectiveness}
In this section, we compare our method with a classical keyword search algorithm BANKS \cite{DBLP:BANKS} over DBLP and Yago to show the effectiveness of our method. Furthermore, since each resource in DBPedia is annotated by Wikipedia documents, so we design a stronger baseline named as ``Annotated SPARQL'' for DBPedia.

\subsubsection{Case Study}

We show six sample queries in Table \ref{table:samplequerydblp}, \ref{table:samplequeryyago} and \ref{table:samplequerydbpedia} for the case study.

\begin{table}[h]
\small
\centering
  \begin{tabular}{|c|c|c|c|}
  \hline
   &   & \multicolumn{2}{c|}{SK Query} \\
  \cline{3-4}
   &  Query Sematic& SPAQRL& Keywords	\\
  \cline{1-4}
 Q3 & \tabincell{p{3cm}}{Which researchers on keyword search published papers in VLDB 2004 and DEXA 2005?}	& \tabincell{l}{Select ?person where $\{$ \\?paper year 2004;\\?paper booktitle VLDB;\\?paper1 year 2005;\\?paper1 booktitle DEXA;\\?paper1 creator ?person;\\?paper dc:creator ?person;$\}$}& \tabincell{p{1cm}}{keyword search}\\
  \cline{1-4}
  Q4& \tabincell{p{3cm}}{Which papers in KDD 2005 about concept-drifting are written by Jiawei Han?} 	& \tabincell{l}{Select ?paper where $\{$ \\?paper year 2003;\\paper booktitle KDD;\\?paper creator ?person;\\?person name Jiawei Han;$\}$} & \tabincell{p{1cm}}{concept-drifting} \\
  \hline
  \end{tabular}
  \caption{ Sample DBLP Queries for Case Study}
  \label{table:samplequerydblp}
\end{table}

\begin{table}[h]
\small
\centering

  \begin{tabular}{|c|c|c|c|}
  \hline
   &   & \multicolumn{2}{c|}{SK Query} \\
  \cline{3-4}
   &  Query Sematic& SPAQRL& Keywords	\\
  \cline{1-4}
  $Q_1$& \tabincell{p{3cm}}{Which actors/actresses played in Philadelphia are mostly related to Academy Award and Golden Globe Award?}	& \tabincell{l}{Select ?p where$\{$\\ ?p type actor;\\?p actedIn ?f;\\?f label ``Philadelphia''; $\}$}& \tabincell{p{1cm}}{Academy Award, \\Golden Globe Award} \\
  \cline{1-4}
  $Q_2$ &\tabincell{p{3cm}}{Which Turing Award winners in the field of database are mostly related to Toronto?} & \tabincell{l}{Select ?p where $\{$ \\?p type scientist;\\?p hasWonPrize ?a;\\?a label ``Turing Award'';$\}$} & \tabincell{p{1cm}}{Toronto, \\database}	\\
  \hline
  \end{tabular}
  \caption{ Sample Yago Queries for Case Study}
  \label{table:samplequeryyago}
\vspace{-0.2in}
\end{table}

\begin{table}[h]
\small
\centering
  \begin{tabular}{|c|c|c|c|}
  \hline
   &   & \multicolumn{2}{c|}{SK Query} \\
  \cline{3-4}
   &  Query Sematic& SPAQRL& Keywords	\\
  \cline{1-4}
   $Q_3$ & \tabincell{p{3cm}}{Which states of Germany are governed by the Social Democratic Party?} & \tabincell{l}{Select ?s where $\{$ \\?s country ?g;\\?g name ``Germany'';$\}$}	&  \tabincell{p{1.8cm}}{Social Democratic Party} \\
  \cline{1-4}
   $Q_4$ & \tabincell{p{3cm}}{Which monarchs of the United Kingdom were married to a German?}	& \tabincell{l}{Select ?u where $\{$ \\?u spouse ?s;\\?s birthPlace ?c;\\?c name ``Germany'';$\}$} & \tabincell{p{1.8cm}}{United Kingdom,\\ monarch} \\
  \hline
  \end{tabular}
  \caption{ Sample QALD Queries over DBPedia for Case Study}
  \vspace{-0.2in}
    \label{table:samplequerydbpedia}
\end{table}

\textbf{DBLP.} Let us consider the two sample queries in DBLP. The top-3 results answered by the SK query and BANKS over DBLP are shown in Table \ref{table:casestudyresdblp}.

$Q_3$:\emph{ Which researchers on keyword search published papers in VLDB 2004 and DEXA 2005?}

The three results returned in SK query are three researchers named ``Kesheng Wu'', ``Jeffrey Xu Yu'' and ``Maurice van Keulen''. All of them published paper in both VLDB 2004 and DEXA 2005. As well, they wrote papers about keyword search before. However, the first three results returned by traditional keyword search are ``Katsumi Tanaka'', ``Mong-Li Lee'' and ``Reda Alhajj''. Although all of them are interested in keyword search, none of them published paper in VLDB 2004.

$Q_4$:\emph{ Which papers in KDD 2005 about concept-drifting are written by Jiawei Han?}

The first result returned in SK query is a paper named``Mining concept-drifting data streams using ensemble classifiers''. This paper was written by Jiawei Han and published in KDD 2005. This is a paper closely related to concept-drifting. This is the best answer to query $Q_2$. The other two results of SK queries are still two paper written by Jiawei Han and published in KDD 2005. In contrast, the first two results returned by traditional keyword search are two papers about concept-drifting, but none of them was published in KDD 2005 or written by Jiawei Han. The third result of traditional keyword search is a researcher, which is more unrelated to the query.

\begin{table}
\small
\centering
  \begin{tabular}{|c|c|c|}
  \hline
  & \tabincell{c}{Top-3 Results of SK Query}  & \tabincell{c}{Top-3 Results of BANKS } \\
  \hline
   & \tabincell{c}{Kesheng Wu}	& \tabincell{l}{Katsumi Tanaka}  \\
  \cline{2-3}
   Q3 & \tabincell{c}{Jeffrey Xu Yu}  & \tabincell{l}{Mong-Li Lee}  \\
  \cline{2-3}
   & \tabincell{c}{Maurice van Keulen}	& \tabincell{l}{Reda Alhajj} \\
  \cline{1-3}
   & \tabincell{c}{Mining concept-drifting data \\streams using ensemble classifiers}	& \tabincell{l}{On Reducing Classifier Granularity \\in Mining Concept-Drifting \\Data Streams.}   \\
  \cline{2-3}
   Q4 & \tabincell{c}{CLOSET+: searching for \\the best strategies for mining  \\frequent closed itemsets.}  & \tabincell{l}{ACE: Adaptive Classifiers-Ensemble \\System for Concept-Drifting \\Environments.}  \\
  \cline{2-3}
   & \tabincell{c}{CloseGraph: mining closed \\frequent graph patterns.}& \tabincell{l}{Baile Shi} \\
  \hline
  \end{tabular}
  \caption{Effectiveness Results for Sample DBLP Queries}
  \label{table:casestudyresdblp}
\end{table}

\textbf{Yago.} Let us consider the two sample queries in Yago. The top-3 results answered by the SK query and BANKS over Yago are shown in Table \ref{table:casestudyresyago}.

$Q_1$:\emph{ Which actors/actresses played in Philadelphia are mostly related to Academy Award and Golden Globe Award?}

We have analyzed query $Q_1$ in Section \ref{sec:introduction}. For comparison, we use keywords $\{actors,$ $actresses,$ $Philadelphia,$ $Academy$ $Award,$ $ Golden$ $Globe$ $Award\}$ for keyword search. Generally, SK query returns more reasonable answers than the traditional keyword search. In contrast, the first two results returned traditional keyword search are ``Grace Kelly'' and ``George Cukor''. Grace Kelly lived in Philadelphia, and George Cukor is also an actor that directed the film, \emph{The Philadelphia Story}, in 1940.

$Q_2$:\emph{ Which Turing Award winners are mostly related to Toronto?}

The first result returned in SK query is ``Stephen Cook''. As we know, Stephen Cook is a professor in University of Toronto. He won the Turing award for his contributions to complexity theory. This is the best answer to query $Q_2$. The second answer is ``William Kahan''. Prof. William Kahan was born in Toronto and won the Turing award for his contributions to the numerical analysis algorithm. The third one is ``Kenneth E. Iverson''. Prof. Kenneth E. Iverson also received the Turing Award. He was died in Toronto.

In contrast, the first two results returned by traditional keyword search are ``English Language'' and ``Princeton University''. Obviously, they are non-informative results. Here, keywords for keyword search that we use are $\{Turing$ $Award,$ $winners,$ $Toronto\}$.

\vspace{-0.1in}
\begin{table}[h]
\small
\centering
  \begin{tabular}{|c|c|c|}
  \hline
  & \tabincell{c}{Top-3 Results of SK Query}  & \tabincell{c}{Top-3 Results of BANKS } \\
    \hline
    \hline
  & \tabincell{c}{Denzel Washington}	& \tabincell{c}{Grace Kelly  }\\
  \cline{2-3}
  $Q_1$ & \tabincell{c}{Joanne Woodward} & \tabincell{c}{George Cukor}	\\
  \cline{2-3}
  & \tabincell{c}{Antonio Banderas	}&\tabincell{c}{ Joanne Woodward }\\
  \cline{1-3}
  & \tabincell{c}{Stephen Cook} 	& \tabincell{c}{English language }\\
  \cline{2-3}
   $Q_2$ & \tabincell{c}{William Kahan }&\tabincell{c}{ Princeton University	 }\\
  \cline{2-3}
   & \tabincell{c}{Kenneth E. Iverson	}&\tabincell{c}{ Turing Award  }\\
  \hline
  \end{tabular}
  \vspace{-0.2in}
  \caption{Effectiveness Results for Sample Yago Queries}

  \label{table:casestudyresyago}
\end{table}
\vspace{-0.1in}

\textbf{DBPedia $\&$ QALD.} Let us consider the two sample QALD queries over DBPedia. The top-3 results answered by BANKS, the SK query and ``Annotated SPARQL'' over DBPedia are shown in Table \ref{table:casestudyresdbpedia}.

\begin{table}[H]
\small
\centering
\begin{threeparttable}
  \begin{tabular}{|p{0.07in}|c|c|c|}
  \hline
  & \tabincell{c}{Top-3 Results of\\ SK Query}  & \tabincell{c}{Top-3 Results of\\ BANKS } & \tabincell{c}{Top-3 Results of\\ Annotated SPARQL} \\
  \hline
    \hline
   & Hanau	& Australia  & Hans-Ulrich Rudel \\
  \cline{2-4}
   $Q_3$ & Hanhofen  & Bombardier Transportation & Hans Dauser\\
  \cline{2-4}
   & Hanover	& Canada& Hans Heidtmann\\
   \cline{1-4}
   & \tabincell{c}{William IV of\\ the United Kingdom}	& \tabincell{c}{2004 Amsterdam\\ Admirals season } & \tabincell{c}{William IV of\\ the United Kingdom}\\
  \cline{2-4}
   $Q_4$ & \tabincell{c}{Carl XVI Gustaf of\\ Sweden}  & \tabincell{c}{2004 Berlin\\ Thunder season} & \tabincell{c}{Beatrix of\\ the Netherlands}\\
  \cline{2-4}
   & \tabincell{c}{Beatrix of\\ the Netherlands}	& \tabincell{c}{2004 Cologne\\ Centurions season}& Switzerland\\
  \hline
  \end{tabular}
  \end{threeparttable}
\vspace{-0.2in}
      \caption{Effectiveness Results over DBPedia for Sample QALD Queries}
  \label{table:casestudyresdbpedia}
\end{table}

$Q_3$:\emph{ Which states of Germany are governed by the Social Democratic Party?}

The first three results returned in SK query are three places in Germany and governed by the Social Democratic Party. However, the first three results returned in annotated SPARQL are three members of the Social Democratic Party in Germany. The first three results returned by traditional keyword search are three place far from Germany. Hence, the results of SK queries are more informative than the other two methods. Here, keywords for keyword search are $\{state,$ $Germany,$ $ govern,$ $Social Democratic Party\}$, which are given in QALD.

$Q_4$:\emph{ Which monarchs of the United Kingdom were married to a German?}

The first result of both SK query and annotated SPARQL are William IV of the United Kingdom, which is the best answer. The other two results of SK query are still two royals in Europe. However, the third result of annotated SPARQL is a European nation. In addition, the first three results returned by traditional keyword search are non-informative results. Here, keywords for keyword search are $\{United$ $Kingdom,$ $monarch,$ $married,$ $German\}$, which are also given in QALD.

\subsubsection{NDCG@k over Yago and DBLP}
In order to quantify the effectiveness of SK query, we evaluate the NDCG (Normalized Discounted Cumulative Gain \cite{DBLP:NDCG}) of both SK query and the keyword search. Since there are no golden standards, we invite 10 volunteers to judge the result quality. Specifically, we ask each volunteer to rate the goodness of the results returned by SK query and the keyword search method. The score is between $1$ and $5$. Higher the score, better the result.

Table \ref{table:NDCG} reports NDCG@k values by varying $k$ from 3 to 10 in both Yago and DBLP. SK query outperforms the traditional keyword search by $20\%$-$50\%$. Furthermore, we find that the gap in Yago is larger than that in DBLP. The reason is that Yago has more complex schema than DBLP.
Thus, keywords may result in more ambiguity in Yago than in DBLP. It means that the superiority of SK query is more pronounced in semantic-rich data.

\begin{table}
\small
\centering
  \begin{tabular}{|c|c|c|c|c|}
  \hline
  & & NDCG@3 & NDCG@5 & NDCG@10\\
  \hline
  Yago & BANKS &0.3455	& 0.39	& 0.4643\\
  \cline{2-5}
  & SK query &0.815 &	0.868	& 0.872 \\
  \hline
  \hline
  DBLP & BANKS &0.7143	&0.684	&0.685\\
  \cline{2-5}
  & SK query &0.93	&0.8867	&0.8738\\
  \hline
  \end{tabular}
 \caption{Average NDCG Values}
 \vspace{-0.2in}
 \label{table:NDCG}
\end{table}

\subsubsection{MAP over DBPedia}

Since QALD provides the standard answers of each queries, we evaluate the MAP (Mean Average Precision \cite{DBLP:conf/sigir/TurpinS06})  to compare the SK query with BANKS and ``Annotated SPARQL''.

\begin{table}[h]
\small
\centering
  \begin{tabular}{|c|c|c|c|}
  \hline
 & BANKS & Annotated SPARQL &SK query\\
  \hline
 MAP  &0.012	& 0.192	 &0.205 \\
  \hline
  \end{tabular}
 \caption{MAP Value over DBPedia $\&$ QALD}
 \label{table:MAP}
 \vspace{-0.2in}
\end{table}

Table \ref{table:MAP} reports MAP values of our ten QALD queries. Both SK query and annotated SPARQL outperforms the traditional keyword search by a order of magnitude. The MAP value of the ``annotated SPARQL'' is smaller than the SK query. This is because that the ``annotated SPARQL'' can do well when the documents associated with the matches contains the keywords. In other words, the ``annotated SPARQL'' can work, only when the relation between the matches and the keywords is explicit. However, in pracitce, the relation between the matches and the keywords is often implicit. Then, the SK query do better.

\begin{table*}
\centering
\begin{tabular}{|c||c|c|c|c|c|}
\nop{
\hline
\multicolumn{6}{|c|}{\textbf{The Number of Graph Matching Operations}}\\
}
\hline
  &$Q_1$&	$Q_2$&	$Q_3$&	$Q_4$&	$Q_5$\\
  \hline
  Naive Backward Search&292268 	&21885	&254674	&872426	&2747\\
  \cline{1-6}
 Our Approach&354 	&5	&1684	&669	&1548\\
  \hline
  \end{tabular}
\caption{The Number of Graph Matching Operations on DBLP}
\vspace{-0.2in}
\label{table:sparqlcndblp}
\end{table*}

\begin{table*}
\centering
\begin{tabular}{|c||c|c|c|c|c|c|c|c|}
\nop{
\hline
\multicolumn{9}{|c|}{\textbf{The Number of Graph Matching Operations}}\\
}
\hline
  &$Q_1$&	$Q_2$&	$Q_3$&	$Q_4$&	$Q_5$&	$Q_6$&	$Q_7$& $Q_8$\\
  \hline
  Naive Backward Search&600283	&563736	&301&	167958	&231210	&271929	&94012	&254848	\\
  \cline{1-9}
  Our Approach&6	&55	&269&	5414	&32	&9	&292	&15	\\
  \hline
  \end{tabular}
  \vspace{-0.25in}
\caption{The Number of Graph Matching Operations on Yago}
\label{table:sparqlcnyago}
\end{table*}

\begin{table*}
\vspace{-0.1in}
\centering
\begin{tabular}{|c||c|c|c|c|c|c|c|c|c|c|}
\nop{
\hline
\multicolumn{11}{|c|}{\textbf{The Number of Graph Matching Operations}}\\
}
\hline
  &$Q_1$&	$Q_2$&	$Q_3$&	$Q_4$&	$Q_5$&$Q_6$&	$Q_7$&	$Q_8$&	$Q_9$&	$Q_{10}$\\
  \hline
  Naive Backward Search&	31083&	136777&19904	&	847&  19454&	40302&5076	&23422	& 16079&2786	 \\
  \cline{1-11}
  Our Approach	&13824	&	3941&	4769  &847	&  40 &	89&18	&	23422&	16079&  1 \\
  \hline
  \end{tabular}
  \vspace{-0.25in}
\caption{The Number of Graph Matching Operations on DBPedia}
\label{table:sparqlcnDBPedia}
\end{table*}

\subsection{Efficiency Study}\label{sec:efficiency}
In this section, we evaluate the efficiency of SK query in large real graphs. Here, the default number of returned results is set to be $10$.

\subsubsection{Offline Performance}
We report the index size and index construction time in Table \ref{table:index}. Since our structural index is based on the efficient sequential pattern mining, we can finish the structural index construction in several minutes.

\vspace{-0.05in}
\begin{table}[h]
\small
\centering
  \begin{tabular}{|c|c|c|}
  \hline
  &Index Construction Time(s) & Index Size(MB)	\\
  \hline
  DBLP & 77.885 & 377.977	 \\
  \hline
  Yago & 176.67 &	844.066	\\
  \hline
  DBPedia & 600.263 & 283.559	 \\
  \hline
  \end{tabular}
 \caption{Index Size and Index Construction time}\label{table:index}
 \vspace{-0.2in}
\end{table}
\vspace{-0.2in}

\subsubsection{Pruning Effect of Structural Index}
Based on the indices introduced in Section \ref{sec:index}, we can avoid many times to call Algorithm \ref{alg:sparqlmatchingnew} for graph matching by pruning many unsatisfied vertices. Moreover, the vertices which are too far to be in a final answer can be safely pruned. In this experiment, we report the pruning efficiency of our structural index. We make a comparison of the number of graph matching operations that the advanced backward search accessed and the number of graph matching operations that the naive backward search accessed.

Tables \ref{table:sparqlcndblp}, \ref{table:sparqlcnyago} and \ref{table:sparqlcnDBPedia} show the number of graph matching operations on DBLP, Yago and DBPedia. The number of graph matching operations in advanced backward search is not less than the basic backward search. In most case, we avoid a large number number of graph matching operations.

\vspace{-0.05in}
\subsubsection{Online Performance}
In this section, we evaluate the efficiency of our method. Figure \ref{fig:efficiency} shows the time cost of the three methods.

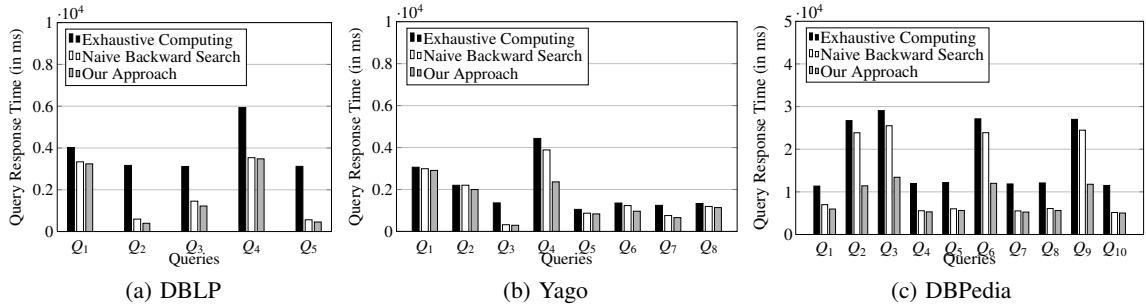
\begin{figure*}%
\vspace{-0.2in}
\subfigure[DBLP]{%
		\resizebox{0.25\columnwidth}{!}{
				\begin{tikzpicture}[font=\Large]
 		 \begin{axis}[
    			major x tick style = transparent,
    			ybar,
    width = 10cm,
               height = 8cm,
    ymin = 1,
    ymax=10000,
   			ymajorgrids = true,
   			ylabel = {Query Response Time (in ms)},
    			xlabel = {Queries},
    			symbolic x coords = {$Q_1$,$Q_2$,$Q_3$,$Q_4$,$Q_5$},
    			scaled y ticks = true,
			bar width=6pt,
			legend pos= north west,
 legend cell align=left
   		]

    \addplot [fill=black] coordinates {($Q_1$, 4022) ($Q_2$, 3165.75) ($Q_3$, 3113.25) ($Q_4$, 5934.5) ($Q_5$, 3120.75)};

\addplot [fill=white] coordinates {($Q_1$, 3338) ($Q_2$, 598) ($Q_3$, 1455) ($Q_4$, 3536) ($Q_5$, 564)};

    		\addplot [fill=black!30] coordinates {($Q_1$, 3240)  ($Q_2$, 397.25) ($Q_3$, 1223) ($Q_4$, 3481) ($Q_5$, 458.25)};

   		 \legend{Exhaustive Computing,Naive Backward Search,Our Approach}
  		\end{axis}
\end{tikzpicture}
		}
       \label{fig:efficiency_dbpedia}%
       }%
   \subfigure[Yago]{%
		\resizebox{0.3\columnwidth}{!}{
				\begin{tikzpicture}[font=\Large]
 		 \begin{axis}[
    			major x tick style = transparent,
    			ybar,
    width = 12cm,
               height = 8cm,
    ymin = 1,
    ymax=10000,
   			ymajorgrids = true,
   			ylabel = {Query Response Time (in ms)},
    			xlabel = {Queries},
    			symbolic x coords = {$Q_1$,$Q_2$,$Q_3$,$Q_4$,$Q_5$,$Q_6$,$Q_7$,$Q_8$},
    			scaled y ticks = true,
			bar width=6pt,
			legend pos= north west,
 legend cell align=left
   		]

    \addplot [fill=black] coordinates {($Q_1$, 3061) ($Q_2$, 2200.75) ($Q_3$, 1360.5) ($Q_4$, 4433.5) ($Q_5$, 1051.25) ($Q_6$, 1353.25) ($Q_7$, 1243.25) ($Q_8$, 1331	)};

\addplot [fill=white] coordinates {($Q_1$, 2990.25) ($Q_2$, 2203.25	) ($Q_3$, 320	) ($Q_4$, 3886) ($Q_5$, 869.5) ($Q_6$, 1233.5) ($Q_7$, 757) ($Q_8$, 1191.5)};

    		\addplot [fill=black!30] coordinates {($Q_1$, 2908.5) ($Q_2$, 1996) ($Q_3$, 293.5) ($Q_4$, 2362.25) ($Q_5$, 834.75) ($Q_6$, 963.25) ($Q_7$, 658.5) ($Q_8$, 1136.75) };

   		 \legend{Exhaustive Computing,Naive Backward Search,Our Approach}
  		\end{axis}
\end{tikzpicture}
		}
       \label{fig:efficiency_yago}%
       }%
\subfigure[DBPedia]{%
		\resizebox{0.3\columnwidth}{!}{
				\begin{tikzpicture}[font=\LARGE]
 		 \begin{axis}[
    			major x tick style = transparent,
    			ybar,
    width = 14cm,
               height = 9.2cm,
    ymin = 1,
    ymax=50000,
   			ymajorgrids = true,
   			ylabel = {Query Response Time (in ms)},
    			xlabel = {Queries},
    			symbolic x coords = {$Q_1$,$Q_2$,$Q_3$,$Q_4$,$Q_5$,$Q_6$,$Q_7$,$Q_8$,$Q_9$,$Q_{10}$},
    			scaled y ticks = true,
			bar width=6pt,
			legend pos= north west,
 legend cell align=left
   		]

    \addplot [fill=black] coordinates {($Q_1$, 11350	) ($Q_2$, 26705   ) ($Q_3$,29044   ) ($Q_4$, 11968	) ($Q_5$, 12184	) ($Q_6$, 27122   ) ($Q_7$, 11853) ($Q_8$, 12089) ($Q_9$, 27005) ($Q_{10}$, 11504)};

\addplot [fill=white] coordinates {($Q_1$, 7000) ($Q_2$, 23863) ($Q_3$, 25504) ($Q_4$, 5571) ($Q_5$, 6027) ($Q_6$, 23874) ($Q_7$, 5535) ($Q_8$, 6087)($Q_9$, 24461)($Q_{10}$, 5164)};

    		\addplot [fill=black!30] coordinates {($Q_1$, 5972) ($Q_2$, 11419) ($Q_3$, 13427) ($Q_4$, 5310) ($Q_5$, 5626) ($Q_6$, 12002) ($Q_7$, 5260) ($Q_8$, 5618)($Q_9$, 11784)($Q_{10}$, 5041) };

   		 \legend{Exhaustive Computing,Naive Backward Search,Our Approach}
  		\end{axis}
\end{tikzpicture}
		}
       \label{fig:efficiency_dbpedia}%
       }%
       \vspace{-0.3in}
 \caption{Online Performance}%
 \label{fig:efficiency}
\end{figure*}

As shown in Figure \ref{fig:efficiency}, our method outperforms the baseline method by 2 or more times in most case. Especially for $Q_3$ on Yago and $Q_2$, $Q_5$ on DBLP, our method only takes a fifth of the exhaustive-computing. This is because that the matches of these SPARQLs are close to vertices containing keywords. Thus, the query processing can terminate soon.

Note that, because our inverted index for keywords are stored in disk, keywords mapping will cost much time and takes up a large part of the total time. Hence, it is difficult for our method to improve the efficiency too much.

\section{Related Work}\label{sec:related}
For SPARQL query, there have been many works to study it, such as \cite{DBLP:Hexastore,DBLP:swstore,DBLP:rdf3x,DBLP:xrdf3x,DBLP:gStore,DBLP:conf/sigmod/BorneaDKSDUB13,DBLP:journals/pvldb/ZengYWSW13}. Some of them \cite{DBLP:swstore,DBLP:conf/sigmod/BorneaDKSDUB13} store the RDF triples into RDBMS and answer the SPARQL via join operations. RDF-3x \cite{DBLP:rdf3x,DBLP:xrdf3x} and Hexastore \cite{DBLP:Hexastore} create indexes for each permutation of subject, predicate and object. Since an RDF dataset can also be modeled as a graph, Trinity.RDF \cite{DBLP:journals/pvldb/ZengYWSW13} and gStore \cite{DBLP:gStore} deem answer the SPARQL in an RDF dataset as finding the subgraph matches over an RDF graph. Trinity.RDF and gStore design a subgraph match algorithm similar to VF2 \cite{DBLP:VF2} to answer the SPARQL query. VF2 \cite{DBLP:VF2} is an early efforts for subgraph isomorphism check. VF2 starts with a vertex and explore to a vertex connected from the already matched query vertices one by one.

For keyword search, existing keyword search techniques over RDF graphs can be classified into the following two categories. The first kind of methods \cite{Q2Semantic:ESWC2008,ICDE09:keyword,keyword:ISWC2011,CoSi:www2011} interpret keywords as SPARQL queries, and then retrieve results by involving existing SPARQL query engines. Another kind of methods aim to find the small-size substructures (in RDF graphs) that contains all keywords. The top-k substructures, such as trees \cite{DBLP:BANKS,DBLP:Bidirectional,DBLP:BLINKS,DBLP:DPBF,DBLP:star,DBLP:journals/tkde/LeLKD14,DBLP:conf/adma/WangZPZ12}, cliques \cite{DBLP:rcliques}, computed on the basis of a scoring function are returned to users.

There are also many approaches to mine some frequent patterns to build indices in graph database \cite{DBLP:gIndex,DBLP:gSpan,DBLP:gaddi}. Among these works, gIndex \cite{DBLP:gIndex} and gSpan \cite{DBLP:gSpan} can be applied to small graphs in a database of multiple graphs, but not support mining patterns in a single graph. GADDI \cite{DBLP:gaddi} tries to finding all
the matches of a query graph in a given large graph, but it can only
support a graph with thousands of vertices while recent RDF data graph
may have hundred thousands of entities.

To the best of our knowledge, although there exist a few previous works \cite{DEBU:sparqlandkeyword,DBLP:CE2} for the hybrid query combined SPARQL
and keyords, there has been no existing work on SK query defined as the above. Elbassuoni et al. \cite{DEBU:sparqlandkeyword} assumes that each RDF triple may have associated text passages. Then, Elbassuoni et al. extend the triple patterns in SPARQL with keyword conditions. Moreover, $CE^2$ \cite{DBLP:CE2} assumes that each resource associate with a document. Then, $CE^2$ extend the variables in SPARQL with keyword conditions. Nonetheless, most current RDF datasets do not provide neither text passages to annotate triples nor documents to annotate resources. In summary, both of these methods cannot handle our example queries. Also, the SK query that we define can apply to most existing RDF datasets.

As well, in \cite{SIGMOD:structuralkeyword}, the authors define a new query language that blends keyword search with structured query processing. \cite{SIGIR:ObjectRetrieval} utilizes some given kinds of SPARQL to improve the result of object retrieval. Moreover, \cite{ESWC08:HS,OTM:GoNTogle} try to extend keyword search with semantics. Zou et al. \cite{DBLP:conf/sigmod/ZouHWYHZ14} translate natural language questions into SPARQL queries.

\section{Conclusions}\label{sec:conclusions}
In this paper, we have proposed a new kind of query (SK query) that integrates SPARQL and keywords. To handle this kind of query, we firstly introduce a basic method based on the backward search. However, this basic solution faces several performance issues. Hence, we build up a structural index. Our structural index is based on frequent star pattern in RDF data. By using the indices, we propose an advanced strategy to deal with SK queries. Finally, with three real RDF datasets, we demonstrate that the our method can outperform the baseline both in effectiveness and efficiency.







\small
\bibliographystyle{abbrv}
\bibliography{sigproc}

\normalsize






\appendix

\section{Queries in Experiments}
\label{sec:allresults}

Table \ref{table:allsamplequeries} shows all of our sample queries over Yago and DBLP. Here, since our institution is in China, most volunteers that we invite are Chinese. Hence, some sample queries are about China.

For more reasonable experiments, so we also sample $10$ non-aggregation QALD queries over DBPedia to evaluate our method. All QALD queries over DBPedia are shown in Table \ref{table:allsamplequeriesdbpedia}.

\begin{table*}
\small
\centering

  \begin{tabular}{|c||c|c|c|c|}
  \hline
  & &   & \multicolumn{2}{c|}{SK Query} \\
  \cline{4-5}
   &   &  Query Sematic& SPAQRL& Keywords	\\
  \cline{1-5}
  &  Q1& \tabincell{p{7cm}}{Which researchers on keyword search published papers in VLDB 2004 and DEXA 2005?}	& \tabincell{p{5cm}}{Select ?person where $\{$ \\?paper year 2004\\?paper booktitle VLDB\\?paper1 year 2005\\?paper1 booktitle DEXA\\?paper1 creator ?person\\?paper creator ?person$\}$}& \tabincell{p{3.5cm}}{keyword search}\\
  \cline{2-5}
  &  Q2& \tabincell{p{7cm}}{Which papers in KDD 2005 about concept-drifting are written by Jiawei Han?} 	& \tabincell{p{5cm}}{Select ?paper where $\{$ \\?paper year 2003\\paper booktitle KDD\\?paper creator ?person\\?person name Jiawei Han$\}$} & \tabincell{p{3.5cm}}{concept-drifting} \\
  \cline{2-5}
  \textbf{DBLP}&  Q3& \tabincell{p{7cm}}{Who wrote a paper in ICDM 2005 with others and knew Tamer?} 	& \tabincell{p{5cm}}{Select ?person1 where $\{$ \\?paper year 2005;\\?paper booktitle ``ICDM'';\\?paper creator ?person1;\\?paper dc:creator ?person2$\}$} & \tabincell{p{3.5cm}}{Tamer} \\
  \cline{2-5}
  &  Q4& \tabincell{p{7cm}}{Who wrote a paper VLDB 2005 and kept a good relationship to Jian Pei and Wen Jin?} 	& \tabincell{p{5cm}}{Select ?person2 where $\{$ \\?paper year ``2005'';\\?paper booktitle ``VLDB'';\\?paper creator ?person2;$\}$} & \tabincell{p{3.5cm}}{Jian Pei,\\ Wen Jin} \\
  \cline{2-5}
  &  Q5& \tabincell{p{7cm}}{Which two researchers did research about skyline and coauthored a paper in VLDB 2005?} 	& \tabincell{p{5cm}}{Select ?person1, ?person2 where $\{$ \\?paper year ``2005'';\\?paper booktitle ``VLDB'';\\?paper dc:creator ?person1;\\?paper dc:creator ?person2$\}$} & \tabincell{p{3.5cm}}{Skyline} \\
  \hline
  \hline
  &   Q1& \tabincell{p{7cm}}{Which actresses played in Philadelphia are mostly related to Academy Award and Golden Globe Award?}	& \tabincell{p{5cm}}{Select ?p where$\{$\\ ?p type actor;\\?p actedIn ?f;\\?f label ``Philadelphia''; $\}$}& \tabincell{p{3.5cm}}{Academy Award, \\Golden Globe Award} \\
  \cline{2-5}
  &   Q2 &\tabincell{p{7cm}}{Which Turing Award winners in the field of database are mostly related to Toronto?} & \tabincell{p{5cm}}{Select ?p where $\{$ \\?p type scientist;\\?p hasWonPrize ?a;\\?a label ``Turing Award'';$\}$} & \tabincell{p{3.5cm}}{Toronto, \\database}	\\
  \cline{2-5}
  &   Q3 &\tabincell{p{7cm}}{Which Microsoft¡¯s products are about SDK?} & \tabincell{p{5cm}}{Select ?c where $\{$ \\?c type company;\\?c label ``Microosft'';\\?c created ?s;$\}$} & \tabincell{p{3.5cm}}{SDK}	\\
  \cline{2-5}
  &   Q4 &\tabincell{p{7cm}}{Which English film producers did act in a Comedy film and relate to Peking University?} & \tabincell{p{5cm}}{Select ?p where $\{$ \\?p actedIn ?f1;\\?p y:created ?f2;\\?f1 type ComedyFilms;\\?f2 y:hasProductionLanguage English;$\}$} & \tabincell{p{3.5cm}}{Peking University }	\\
  \cline{2-5}
 \textbf{Yago} &   Q5 &\tabincell{p{7cm}}{Which top members of Communist Party of China are related Kissinger?} & \tabincell{p{5cm}}{Select ?p where $\{$ \\?p isAffiliatedTo ?u;\\?u label ``Communist Party of China'';\\?p type Politician;$\}$} & \tabincell{p{3.5cm}}{Kissinger}	\\
  \cline{2-5}
  &   Q6 &\tabincell{p{7cm}}{Whose father was United States Army generals and took part in Normandy Invasion?} & \tabincell{p{5cm}}{Select ?p1 where $\{$ \\?p hasChild ?p1;\\ ?p type UnitedStatesArmyGenerals ;$\}$} & \tabincell{p{3.5cm}}{Normandy Invasion}	\\
  \cline{2-5}
  &  Q7 &\tabincell{p{7cm}}{Which state generated a Los Angeles Lakers player that relate to Eagle, Colorado?} & \tabincell{p{5cm}}{Select ?p where $\{$ \\?p bornIn ?c;\\ ?c locatedIn ?s;\\?p type LosAngelesLakersPlayers;$\}$} & \tabincell{p{3.5cm}}{Eagle Colorado}	\\
  \cline{2-5}
  &  Q8 &\tabincell{p{7cm}}{Which participants of People's National Congress did graduate from universities in Beijing?} & \tabincell{p{5cm}}{Select ?p where $\{$ \\?p graduatedFrom ?u;\\?u type UniversitiesInBeijing;$\}$} & \tabincell{p{3.5cm}}{ People's National Congress}	\\

\hline
  \end{tabular}
  \caption{Sample Queries over Yago and DBLP}
  \label{table:allsamplequeries}
\end{table*}

\begin{table*}
\small
\centering
  \begin{tabular}{|c||c|c|c|c|}
  \hline
  & &   & \multicolumn{2}{c|}{SK Query} \\
  \cline{4-5}
  & &  Query Sematic& SPAQRL& Keywords	\\
  \cline{1-5}
  & Q1 & \tabincell{p{6.5cm}}{Which states of Germany are governed by the Social Democratic Party?} & \tabincell{p{5cm}}{Select ?s where $\{$ \\?s country ?g;\\?g name ``Germany'';$\}$}	&  \tabincell{p{3.5cm}}{Social Democratic Party} \\
  \cline{2-5}
  & Q2 & \tabincell{p{6.5cm}}{Which monarchs of the United Kingdom were married to a German?}	& \tabincell{p{5cm}}{Select ?u where $\{$ \\?u spouse ?s;\\?s birthPlace ?c;\\?c name ``Germany'';$\}$} & \tabincell{p{3.5cm}}{United Kingdom,\\ monarch} \\
  \cline{2-5}
  & Q3 & \tabincell{p{6.5cm}}{Which capitals in Europe were host cities of the summer olympic games?}	& \tabincell{p{5cm}}{Select ?u where $\{$ \\?s type Country;\\?s capital ?u;$\}$} & \tabincell{p{3.5cm}}{Olympic games,\\Europe} \\
  \cline{2-5}
  & Q4 & \tabincell{p{6.5cm}}{Who produced films starring Natalie Portman?}	& \tabincell{p{5cm}}{Select ?p where $\{$ \\?f type Film;\\?f producer ?p;$\}$} & \tabincell{p{3.5cm}}{Natalie Portman} \\
  \cline{2-5}
 \tabincell{p{1.2cm}}{\textbf{DBPedia $\&$QALD}} & Q5 & \tabincell{p{6.5cm}}{In which films did Julia Roberts as well as Richard Gere play?}	& \tabincell{p{5cm}}{Select ?f where $\{$ \\?f	type Film;\\?f starring ?p;\\?p name ``Roberts, Julia''$\}$} & \tabincell{p{3.5cm}}{Richard Gere} \\
  \cline{2-5}
  & Q6 & \tabincell{p{6.5cm}}{List all episodes of the first season of the HBO television series The Sopranos!}	& \tabincell{p{5cm}}{Select ?u where $\{$ \\?s	name ``The Sopranos'';\\?u series ?s;$\}$} & \tabincell{p{3.5cm}}{HBO,\\first} \\
  \cline{2-5}
  & Q7 & \tabincell{p{6.5cm}}{In which films directed by Garry Marshall was Julia Roberts starring?}	& \tabincell{p{5cm}}{Select ?f where $\{$ \\?f type Film;\\?f director ?p;\\?p name ''Marshall, Garry'';$\}$} & \tabincell{p{3.5cm}}{Julia Roberts} \\
  \cline{2-5}
  & Q8 & \tabincell{p{6.5cm}}{Which software has been developed by organizations founded in California?}	& \tabincell{p{5cm}}{Select ?u where $\{$ \\?c type Organisation;\\?u developer ?c;\\?u type Software;$\}$} & \tabincell{p{3.5cm}}{California} \\
  \cline{2-5}
  & Q9 & \tabincell{p{6.5cm}}{Which U.S. states possess gold minerals?}	& \tabincell{p{5cm}}{Select ?s where $\{$ \\?s type Place;\\?s country ?g;\\?g name ``the United States'';$\}$} & \tabincell{p{3.5cm}}{gold,\\mineral} \\
  \cline{2-5}
  & Q10 & \tabincell{p{6.5cm}}{Which countries in the European Union adopted the Euro?}	& \tabincell{p{5cm}}{Select ?u where $\{$ \\?u type Country;\\?u ethnicGroup European Union;$\}$} & \tabincell{p{3.5cm}}{Euro} \\
   \hline
  \end{tabular}
  \caption{ Sample QALD Queries over DBPedia}

  \label{table:allsamplequeriesdbpedia}
\end{table*}

\end{document}